\documentclass[twocolumn,aps,prx,amsmath,amssymb,superscriptaddress]{revtex4-1}
\usepackage{graphicx,bm,dsfont,amsmath,amssymb}
\usepackage{soul}
\usepackage{color}
\usepackage{amsthm}
\usepackage{mathrsfs}
\usepackage{subfigure}
\usepackage{graphicx}
\usepackage{float}
\usepackage[percent]{overpic}
\usepackage{lipsum}
\usepackage[colorlinks=true,citecolor=blue,urlcolor=blue]{hyperref}
\normalsize
\usepackage[ruled,vlined, linesnumbered]{algorithm2e}

\begin{document}
\title{Quantum Algorithm for Vector Set Orthogonal Normalization and Matrix QR Decomposition with Polynomial Speedup}
\author{Zi-Ming Li}
\affiliation{School of Integrated Circuits, Tsinghua University, Beijing 100084, China}
\author{Yu-xi Liu}\email{yuxiliu@mail.tsinghua.edu.cn}
\affiliation{School of Integrated Circuits, Tsinghua University, Beijing 100084, China}
\affiliation{Frontier Science Center for Quantum Information, Beijing, China}

\date{\today}

\begin{abstract}
Vector set orthogonal normalization and matrix QR decomposition are fundamental problems in matrix analysis with important applications in many fields. We know that Gram-Schmidt process is a widely used method to solve these two problems. However, the existing methods, including Gram-Schmidt process have problems of high complexity, scaling $O(N^3)$ in the system dimension $N$, which leads to difficulties when calculating large-scale or ill-conditioned problems. With the development of quantum information processing, a series of quantum algorithms have been proposed, providing advantages and speedups over classical algorithms in many fields. In this paper, we propose quantum algorithms to solve these two problems based on the idea of Gram-Schmidt process and quantum phase estimation. The complexity of proposed quantum algorithms is also theoretically and numerically analyzed. We find that our algorithms provide polynomial acceleration over the best-known classical and quantum algorithms on these two problems, scaling $O(N^2\mathrm{poly}(\log N))$ in the dimension $N$ of the system.
\end{abstract}
\maketitle

\section{Introduction}
Vector set orthogonal normalization and matrix QR decomposition stand as fundamental problems in linear algebra and matrix analysis~\cite{horn2012matrix, bhatia2013matrix}. They are also crucial for a myriad of applications in various scenarios, including scientific computing~\cite{businger1965linear, leon2013gram, parlett2000qr}, machine learning and artificial intelligence~\cite{zheng2004real, wang2016unsupervised, aizenberg2012modified}, and many other engineering fields~\cite{zhang2005equal, cheng2016fast, srinivasa2012use}. The goal of vector set orthogonal normalization is to transform sets of vectors into orthogonal normalized ones. Matrix QR decomposition aims at decomposing full-rank matrices into the product of orthogonal matrix $Q$ and upper triangular matrix $R$.  In the last few decades, many classical algorithms have been proposed and improved for solving these two problems~\cite{goodall199313, giraud2005loss, giraud2005rounding, gentleman1975error} by using Gram-Schmidt orthogonalization process. However, these algorithms have problems of relatively high complexity, which scales $O(N^3)$ with the dimension $N$ of the system. Such high computational complexity often brings challenges to managing large-scale or ill-conditioned matrices efficiently and accurately.

In the 1980s, Feynman proposed quantum computers, which use quantum physical systems to achieve information-storage, transmission, and processing~\cite{feynman1985quantum}. With recent advancements in quantum computation, there has been a surge of theoretical and experimental studies in utilizing quantum algorithms to enhance the efficiency of classical computational tasks~\cite{nielsen2010quantum, preskill2018quantum, gibney2019hello, zhong2020quantum, monz2016realization, shor1994algorithms}. In particular, many quantum machine learning-based algorithms have emerged as powerful tools to solve complex linear algebra problems, offering exponential and polynomial speedups to traditional computing methods~\cite{prakash2014quantum}. These quantum algorithms can be applied to solve linear equations~\cite{harrow2009quantum, childs2017quantum}, support vector machine~\cite{rebentrost2014quantum}, principal component analysis~\cite{lloyd2014quantum}, Bayesian network~\cite{low2014quantum}, quantum Boltzmann machine~\cite{amin2018quantum}, and many other problems~\cite{prakash2014quantum}.

Recently, a quantum Gram-Schmidt algorithm was proposed for vector set orthogonal normalization~\cite{zhang2021quantum} using QRAM model~\cite{giovannetti2008quantum} with a query complexity of $O(r^{27}\kappa^{14r})$. The column vectors to be orthogonalized were lined into a matrix with the rank $r$ and conditional number $\kappa$. The proposed algorithm achieves quantum speedup with a low-rank matrix with efficient state readout but reaches high complexity when $r$ and $\kappa$ are larger. Also, a quantum algorithm was proposed for QR decomposition of square matrices with $O(N^{2.5}\mathrm{poly} \log_2(N)/\epsilon^2)$ computational complexity~\cite{ma2020quantum}, where $N$ is the size of the matrix and $\epsilon$ is the desired precision. The scaling on the size of the matrix $N$ successfully achieves polynomial speedup over classical algorithms but the speedup is limited.

Stimulated by previous studies~\cite{zhang2021quantum,giovannetti2008quantum,ma2020quantum}, we here revisit and propose new quantum algorithms for vector set orthogonal normalization and matrix QR decomposition. We explain the details of the proposed algorithms and prove their correctness by theoretical derivations and numerical simulations. We use the QRAM model for efficient quantum initial state preparation~\cite{giovannetti2008quantum, giovannetti2008architectures}, which is a reasonable quantum oracle model widely used. We also analyze the complexity of our algorithm, including the number of quantum gates and the number of oracles used in the algorithms. The query complexity of our vector set orthogonal normalization algorithm scales $O(N^2)$ in the system dimension $N$. Our algorithm provides polynomial speedup over the previous result~\cite{zhang2021quantum}, which scales at least $O(N^{27})$ for the full rank matrix. The query complexity of our matrix QR decomposition algorithm scales $O(N^2\log_2 N)$ in the system dimension $N$, which also provides polynomial speedup over previous result~\cite{ma2020quantum} whose scaling in system dimension $N$ is $O(N^{2.5}\mathrm{poly} \log_2(N))$. Thus, the complexity of the proposed algorithms is optimal to date.

The paper is organized as follows. In Sec.~\ref{section2}, we provide a formal definition of the problems to be solved and briefly review the classical algorithms for vector set orthogonal normalization and QR decomposition problems. We also explain how Gram-Schmidt process can be used to solve these problems. For the completeness of the paper, in Sec.~\ref{section3}, we summarize the main result of the quantum phase estimation algorithm, explain the oracle used in our paper, and give the details on the controlled Hamiltonian simulation step. In Sec.~\ref{section4}, we introduce our quantum algorithm for vector set orthogonal normalization problem and evaluate the performance of the proposed algorithm. In Sec.~\ref{section5}, we introduce our quantum algorithm for QR decomposition and also evaluate the performance of the proposed algorithm. In Sec.~\ref{section6}, we apply our algorithms to several important problems, including linear regression, solving linear equations, and finding eigenvalues. The potential applications of our algorithms are also discussed. In Sec.~\ref{section7}, we summarize our results.

\section{Problem definition and classical algorithms}
\label{section2}
Formally, the problems of vector set orthogonal normalization and matrix QR decomposition of matrix are defined respectively as follows~\cite{horn2012matrix,bhatia2013matrix}.
\theoremstyle{definition}
\newtheorem{definition}{Problem}
\begin{definition}
	\label{define1}
	\textbf{Vector Set Orthogonal Normalization}\\
	\indent Let $S$ be a set containing $M$ elements, each of the element is an $N$-dimensional vector, i.e. $S=\{a_1,a_2,\cdots,a_{M}\}$, \; $a_m\in \mathbb{C}^N, \forall m\in\{1,2,\cdots,M\}$. Find a set of vectors $S'=\{u_1,u_2,\cdots,u_{T}\}$ satisfying:\\
\indent 1.\;$u_{t_1}^{\dagger}u_{t_2}=\delta_{t_1t_2}$,\;$\forall t_1,t_2=1,2,\cdots,T$\\
\indent 2.\;$span\{a_1,a_2,\cdots,a_{M}\}=span\{u_1,u_2,\cdots,u_{T}\}$
\end{definition}
\begin{definition}
	\label{define2}
	\textbf{Matrix QR Decomposition}\\
	\indent Let $A\in \mathbb{C}^{N\times M}$ be an arbitrary matrix with full rank satisfying $N\geq M$. Find an orthogonal matrix $Q\in \mathbb{C}^{N\times M}$ and an upper triangular matrix $R\in \mathbb{C}^{M\times M}$ satisfying:\\
	\indent 1.\;$Q^{\dagger}Q=I_{M\times M}$\\
	\indent 2.\;$A=QR$
\end{definition}

There are many classical numerical methods for Problems~\ref{define1} and~\ref{define2}. Problem~\ref{define2} can be regarded as a generalization of Problem~\ref{define1}. The transformation matrix $R$ is obtained at the same time when the column vectors of given matrix $A$ is orthogonal normalized. One of the most common ways to solve  Problems~\ref{define1} and \ref{define2} is the Gram-Schmidt orthogonalization process. Gram-Schmidt process based algorithms for these two problems are summarized in following Algorithm~\ref{alg1} and Algorithm~\ref{alg2}, respectively.
\begin{algorithm}[H]
	\caption{Gram-Schmidt Process for Vector Set Orthogonal Normalization~\cite{horn2012matrix,bhatia2013matrix}}
    \SetAlgoLined
	\label{alg1}
	\KwIn{$S=\{a_1,a_2,\cdots,a_{M}\}$, $a_m\in \mathbb{C}^N$,$\forall m=1,2,\cdots,M$}
	\KwOut{$S'=\{u_1,u_2,\cdots,u_{T}\}$, satisfying $u_{t_1}^{\dagger}u_{t_2}=\delta_{t_1t_2}$, $\forall t_1, t_2=1,2,\cdots,T$, $span\{a_1,a_2,\cdots,a_{M}\}=span\{u_1,u_2,\cdots,u_{T}\}$}
	\BlankLine
	$v\gets a_1/\parallel a_1\parallel$\;
	$S'\gets \{v\}$\;
		\For{$m\leftarrow2$ \textbf{to} $M$}{
		 $v\gets a_m$;\\
		\For{$u_t$ \textbf{in} $S'$}{
		$v\gets v-u_t^{\dagger} a_mu_t$;
	}
		\If{$\parallel v\parallel>0$}{
		$S'\gets S'\cup\{v/\parallel v\parallel\}$;
	}
	}
\end{algorithm}
\begin{algorithm}[H]
	\caption{Gram-Schmidt Process for Matrix QR Decomposition~\cite{horn2012matrix,bhatia2013matrix}}
	\label{alg2}
\SetAlgoLined
		\KwIn {$A=\left(a_1,a_2,\cdots,a_{M}\right)$, $a_i\in \mathbb{C}^N$, $N\geq M$}
		\KwOut {$Q=(q_1,q_2,\cdots,q_M)$,$R=[R_{m_1m_2}]$ satisfying $QR=A$, $Q^{\dagger}Q=I_{N\times N}$, $R_{m_1m_2}=0$, $\forall m_1> m_2$}
		\BlankLine
	 $q_1\gets a_1/\parallel a_1\parallel$\;
	$R_{11} \gets \parallel a_1\parallel $\;
		\For {$m_1=2$ \textbf{to} $M$}{
		$q_{m_1}\gets a_{m_1}$\;
		\For{$m_2=1$ \textbf{to} $m_1-1$}{
		$R_{m_2m_1} \gets q_{m_2}^{\dagger} a_{m_1}$\;
		$q_{m_1}\gets q_{m_1}-R_{m_2m_1}q_{m_2}$\;
	}
		$R_{m_1m_1} \gets \parallel q_{m_1}\parallel $\;
		$q_{m_1} \gets q_{m_1}/R_{m_1m_1}$\;
	}
\end{algorithm}
Hereafter $\parallel \cdot\parallel $ denotes 2-norm for matrices and vectors unless specified otherwise. In addition to the Gram-Schmidt orthogonalization process, many other classical numerical methods have also been proposed to solve these two problems, among which Householder transformation and Givens transformation are most commonly used. The complexity and numerical stability of these methods has been analyzed in detail over the last few decades~\cite{goodall199313, giraud2005loss, giraud2005rounding, gentleman1975error}.

It is noted that classical algorithms have problems of relatively high complexity, always scaling $O(N^3)$ in the system dimension $N$. For instance, for the QR decomposition of $N\times N$ matrices, the classical Gram-Schmidt process needs $2N^3$ floating point operations, while Householder transformation and Givens transformation need $8N^3/3$ floating point operations. In engineering,  matrices to be processed often have a scale of at least $10^3$. If we do matrix QR decomposition of $10^3\times10^3$ matrix using the classical Gram-Schmidt orthogonalization method, we need to do $2\times 10^9$ floating-point operations. The number of floating-point operations performed by a personal computer per second is about $10^6$, so the running time of the algorithm is about $2000$ seconds, more than half an hour. Thus the $O(N^3)$ scaling in time complexity is usually considered relatively high. To reduce this complexity, we develop quantum algorithms for Problems~\ref{define1} and~\ref{define2} and describe our algorithms in the following sections.

\section{Quantum Phase Estimation}
\label{section3}
Our algorithm is based on quantum phase estimation (QPE), which mainly includes initial state preparation, controlled-$U^{2^j}$ operations, and inverse quantum Fourier transform. In our paper, the initial state preparation is realized via the model of the quantum random-access memory (QRAM)~\cite{lloyd2014quantum, harrow2009quantum, rebentrost2014quantum}. The controlled-$U^{2^j}$ operations are implemented via the Hamiltonian simulation with qubitization~\cite{low2019hamiltonian}. For the completeness of the paper, we here briefly summarize the main results of QPE, the QRAM model, and the qubitization, which are used in our algorithms.

QPE is a widely-used quantum algorithm based on quantum Fourier transformation and is performed as a subroutine in many other quantum algorithms~\cite{nielsen2010quantum,harrow2009quantum,shor1994algorithms}. Several quantum algorithms with QPE as a subroutine have been proved to have exponential acceleration over classical algorithms for many important problems such as discrete logarithm problem, hidden subgroup problem, large integer prime factorization problem, etc~\cite{shor1994algorithms,shor1999polynomial,hales2000improved}.

We assume that an $N$-dimensional unitary operator $U$ has an eigenvalue ${e^{2\pi i\phi_{n}}}$ with unknown $\phi_{n}$ corresponding to eigenstate $|u_{n}\rangle$ for $n=1,\cdots, N$. The goal of the QPE algorithm is to estimate each $\phi_{n}$ with the assistance of black box oracle for preparing an initial state $|u\rangle_s$ and performing the controlled-$U^{2^j}$ operator~\cite{nielsen2010quantum}. Thus, the QPE procedure uses two registers. The first register containing $j$ qubits is initialized in the ground state $|0\rangle^{\otimes j}_{f}$, and the second one is initialized in a state $|u\rangle_s$ in the $N$-dimensional space.  The corresponding quantum circuit for QPE is given in Fig.~\ref{fig1}.

\begin{figure}[h]
   \centering
   	\includegraphics[width=\linewidth]{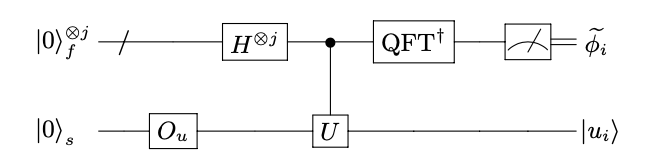}
	\caption{Quantum circuit for QPE. $O_u$ represents the preparation of an initial state $|u\rangle$. In our paper, it is implemented by the QRAM model. The controlled-$U$ operator denotes a series of controlled-$U^{2^j}$ operators $\{C-U^{2^{j-1}},C-U^{2^{j-2}},\cdots,C-U^{2^{1}},C-U^{2^{0}}\}$. ${\rm QFT}^{\dagger}$ denotes inverse quantum Fourier transform. The control qubits are in the first register, labelled by the subscript $f$. The target qubits are in the second register, labelled by the subscript $s$. To avoid confusion, $H$ represents Hadamard gate here in the circuit.}
	\label{fig1}
\end{figure}

The unitary operator $U$ is usually realized via the time evolution operator of a given Hamiltonian $H$, which has spectral decomposition $H=\sum_{n=1}^{N}\lambda_n |u_n\rangle\langle u_n|$. That is,  the unitary  operator $U$ can be expressed as
\begin{equation}
	U=\sum_{n=1}^{N}\exp\left(-i\lambda_n t\right) |u_n\rangle\langle u_n|,
\end{equation}
 in which the phase factor $-\lambda_nt$ corresponds to $2\pi \phi_n$, i.e., $-\lambda_nt=2\pi \phi_n$. Hereafter, we take $\hbar=1$. The second register is initialized in an arbitrary quantum state $|u\rangle_s$ with
		\begin{equation}
			|u\rangle_{s} = \sum_{n=1}^{N}\langle u_n|u\rangle_{s} |u_n\rangle,
		\end{equation}
 which is a linear combination of the eigenstates of $U$. The final state of the quantum circuit before measurements is
		\begin{equation}
			\sum_{n=1}^{N}\langle u_n|u\rangle_{s}\left|\frac{-\lambda_n t}{2\pi}2^j\right\rangle|u_n\rangle,
		\end{equation}
	i.e.,
	\begin{equation}
			\sum_{n=1}^{N}\langle u_n|u\rangle_{s}\left|\phi_n 2^j\right\rangle|u_n\rangle.
\end{equation}
Hereafter $|\beta\rangle|\gamma\rangle$ denotes that the first (second) register is in the state $|\beta\rangle$ ($|\gamma\rangle$).
It is noted that $\lambda_n$ can be $0$. Derivation of QPE  is given in Appendix~\ref{ap1} when the qubit number of the first register is $1$, which is always the case in our proposed algorithms.

A key component in QPE algorithm is to implement an oracle, which is used for preparing the initial state $|u\rangle_s$ and performing the controlled-$U^{2^j}$ operator. Here,  we assume that the state $|u\rangle$ preparation is realized via the QRAM model, which is  used in many quantum algorithms~\cite{lloyd2014quantum, harrow2009quantum, rebentrost2014quantum}. The QRAM model for realizing $O_{u}$ in Fig.~\ref{fig1} can prepare initial state  $|u\rangle_s$  in $O(\log_2 N)$ time, i.e.,
		\begin{equation}
			\label{eq31}
			|0\rangle_s \stackrel{\rm QRAM}{\longrightarrow} O_u|0\rangle_s=|u\rangle_s,
		\end{equation}
with the ground state $|0\rangle_{s}$ of the second register. The implementation of the controlled-$U^{2^j}$ operators requires the Hamiltonian simulation, which is the original intention of Feynman's quantum computer~\cite{feynman1985quantum}.
In our algorithm,  the Hamiltonian simulation is realized via the qubitization~\cite{low2019hamiltonian}, which is considered as the best Hamiltonian simulation method~\cite{miessen2023quantum} up to now. In qubitization, additional auxiliary qubits are introduced to the second register. The Hamiltonian $H$ is accessed through two operators $V$ and $G$, in which the operator $G$ only acts on the states of auxiliary qubits in the second register, and the operator $V$ acts on the states of both auxiliary qubits and original qubits of the second register.  $G$ and $V$ are chosen to satisfy the following condition,
\begin{equation}\label{Eq2-2}
	\left(\langle 0|_aG^{\dagger}\otimes I_s\right)V\left(G|0\rangle_a\otimes I_s\right)\propto H,
\end{equation}
with the ground state $|0\rangle_{a}$ of the  auxiliary qubits. Here, the subscript $a$ is for ancillary qubits introduced in the second register and $s$ is for original qubits of the second register.

In the most common case and here, the Hamiltonian $H$ is assumed to have a form of linear combination of unitaries $H=\sum_{l=1}^{d}\alpha_lV_l$. Then operators  $V$ and $G$ can be chosen as
\begin{align}
	\label{eq5}
	& V=\sum_{l=1}^{d}|l\rangle_a\langle l|_a \otimes V_l,\\\nonumber
	& G=\sum_{l=1}^{d} \sqrt{\frac{\alpha_l}{\sum_{l=1}^{d}|\alpha_l|}} |l\rangle_a\langle 0|_a+\dotsb.
\end{align}

The greatest advantage of the qubitization is that the qubitization has low computational complexity. From Corollary $16$ in~\cite{low2019hamiltonian}, the query complexity for simulating the Hamiltonian $H$ with form of linear combination of unitaries can be given in the following lemma.
\newtheorem{lemma}{Lemma}
\begin{lemma}[the Linear-Combination-of-Unitaries (LCU) algorithm for Qubitization~\cite{low2019hamiltonian}]
	\label{lm1}
Given $H$ is  accessed via operators $V$ and $G$ as in Eq.~(\ref{eq5}), which specifies a Hamiltonian $H=\sum_{l=1}^{d}\alpha_lV_l$, the time evolution by $H$ can be simulated for time $t$ and error $\epsilon_0$ with $O(\alpha t+\log_2(1/\epsilon_0))$ queries to $V$ and $G$, where $\alpha=\sum_{l=1}^{d}|\alpha_l|$. The desired qubit number of the second register, including both original and auxiliary qubits, is $\lceil \log_2 N\rceil+\lceil \log_2 d\rceil+2$ for simulating the Hamiltonian $H$, where $N$ is the dimension of the Hamiltonian $H$. The number of additional two-qubit quantum gates needed is $O\left(\log_2 d\left(\alpha t+\log_2(1/\epsilon_0)\right)\right)$.
\end{lemma}

Thus, all parts of the QPE circuit used for our algorithms are explained. Below,  we give our main algorithms based on the QPE circuit.

\section{Quantum Gram-Schmidt orthogonalization algorithm}
\label{section4}
In this section, based on QPE and the classical Gram-Schmidt orthogonalization algorithm, we propose quantum Gram-Schmidt orthogonalization algorithm to solve vector set orthogonal normalization problem as described in Problem~\ref{define1}.

\subsection{Algorithm Description}
Our algorithm comes from QPE algorithm. We assume that the unitary $U$ as in Fig.~$\ref{fig1}$ is the evolution operator of given Hamiltonian $H$, which has spectral decomposition
\begin{equation}\label{Eq8}
H=\sum_{n=1}^{k}\lambda_n |u_n\rangle\langle u_n|
\end{equation}
with $k<N$ and $\lambda_n>0$ for $\forall n=\{1,2,\cdots,k\}$, i.e., the Hamiltonian has $N-k$ zero eigenvalues. Let us show how $\{|u_1\rangle,|u_2\rangle,\cdots,|u_{k}\rangle\}$ can be expanded to a set of orthogonal complete vectors $\{|u_1\rangle,|u_2\rangle,\cdots,|u_{N}\rangle\}$ via the QPE algorithm. We can write the input state $|u\rangle_s$ prepared by the QRAM model in the second register as the linear combination of these orthogonal complete vectors, i.e.,
\begin{equation}
	\label{eq10}
	|u\rangle_s=\sum_{n=1}^{N}b_n|u_n\rangle=\sum_{n=1}^{k}b_n|u_n\rangle + \sum_{n=k+1}^{N}b_n|u_n\rangle.
\end{equation}
Then, the output state of the QPE algorithm is
\begin{equation}\label{Eq11}
	\sum_{n=1}^{k}b_n|\tilde{\lambda_n}\rangle|u_n\rangle + \sum_{n=k+1}^{N} b_n|0\rangle|u_n\rangle,
\end{equation}
which is after the inverse quantum Fourier transform and before the measurement.  Here, $|\tilde{\lambda_n}\rangle$ and $|0\rangle$ denote the state of the first register. $|u_n\rangle$ is the state of the second register.

We then make measurement in the first register for  the state in Eq.~(\ref{Eq11}) in the computational basis. If the outcome is $0$, then the state of the second register will collapse into
\begin{equation}
	\label{eq7}
	|\psi\rangle = \frac{\sum_{n=k+1}^{N} b_n|u_n\rangle}{\parallel \sum_{n=k+1}^{N} b_n|u_n\rangle\parallel}.
\end{equation}
It is clear that $|\psi\rangle$ is strictly orthogonal with $|u_n\rangle$ for $\forall n=\{1,2,\cdots,k\}$ as
\begin{equation}
	\langle u_{n}|u_{n'}\rangle=0, \;\;\forall n'> k.
\end{equation}
Thus,  a quantum state $|\psi\rangle$, which is orthogonal with arbitrary state $|u_n\rangle$ with $n\leq k$, is constructed. After we obtain $|\psi\rangle$, the input state $|u\rangle$ in Eq.~(\ref{eq10}) can be represented by linear combinations of $|\psi\rangle$ and $\{|u_1\rangle,|u_2\rangle,\cdots,|u_{k}\rangle\}$. Then we can run the QPE algorithm again to find another state which is orthogonal to the states  $|\psi\rangle$ and $\{|u_1\rangle,|u_2\rangle,\cdots,|u_{k}\rangle\}$. This process can be used to construct the bases in Gram-Schmidt orthogonalization procedure. When the first $k$ bases $\{|u_1\rangle,|u_2\rangle,\cdots,|u_{k}\rangle\}$ are known, the $\left(k+1\right)$th base $|u_{k+1}\rangle$ can be defined as $|\psi\rangle$ as in Eq.~(\ref{eq7}). Thus, we construct the $\left(k+1\right)$th base by using $k$ former constructed bases and the input state $|u\rangle$.

Based on the principle above, let us now study how to construct orthogonal normalized bases by using a set of $N$-dimensional vectors  $S=\{a_1,a_2,\cdots, a_{M}\}$ as in Problem \ref{define1}. We find that one qubit in the first register is enough to realize our algorithm, thus the qubit number in the first register is always taken as one in the following description.  In our algorithm,  $N$ components $a_{nm}$ with $n=1,\cdots, N$ of each $N$-dimensional vector $a_{m}$  is encoded in a quantum state $|a_m\rangle$ via the computational basis $|n^{\prime}\rangle\equiv\{0,1\}^{\otimes \log_2N}$ of $\log_2N$ qubits in the second register as
\begin{equation}\label{Eq13}
|a_m\rangle=\frac{1}{{\parallel a_m\parallel}}\sum_{n^{\prime}=0}^{N-1}a_{n^{\prime}m}|n^{\prime}\rangle\equiv\frac{1}{{\parallel a_m\parallel}}\sum_{n=1}^{N}a_{nm}|n-1\rangle.
\end{equation}
It is clear that $|n-1\rangle\equiv |n^{\prime}\rangle$ for $n=1,\cdots, N$.

For the case that $\{a_1,a_2,\cdots, a_{M}\}$ is a set of linearly independent vectors,  orthogonal normalized bases $\{u_{1},\cdots, u_{M} \}$ can be constructed as follows. We first set $|u_{1}\rangle\equiv |a_1\rangle$. Then $a_{1}$ is encoded by $|u_1\rangle$, which corresponds to a normalized base $u_{1}$. The $n$th component $u_{n1}$ of the vector base $u_{1}$ is $u_{n1}\equiv\langle n-1 |u_{1}\rangle$ with $n=1,\cdots, N$.  Based on $|u_{1}\rangle$, we can successively construct $|u_2\rangle,\cdots,|u_M\rangle$ with $\langle u_{i}|u_{j}\rangle=\delta_{ij}$ by using $|a_1\rangle, \cdots, |a_{M}\rangle$  as follows. Suppose $\{|u_1\rangle,|u_2\rangle,\cdots,|u_{k}\rangle\}$ has been constructed and encode the orthogonal normalized bases $\{u_{1},\cdots, u_{k} \}$, with $u_{n_1}^{\dagger}u_{n_2}=\delta_{n_1n_2}$, $\forall n_1,n_2\leq k$ and  $span\{u_{1},\cdots, u_{k} \}\equiv span\{a_1,a_2,\cdots,a_{k}\}$. That is, $k$ orthogonal normalized bases $\{u_{1},\cdots, u_{k} \}$ have been constructed. Let us now construct the $(k+1)$th base $u_{k+1}$. Following the circuit of QPE,  we assume that the input state $|u\rangle_{s}$ of the second register is the state $|a_{k+1}\rangle$, i.e., $|u\rangle_{s}=|a_{k+1}\rangle$, which encodes the $N$-dimensional vector $a_{k+1}$ as shown in Eq.~(\ref{Eq13}). The Hamiltonian $H$ to obtain $u_{k+1}$ can be chosen as $H=\sum_{n=1}^{k}|u_n\rangle\langle u_n|$, which can be simulated as shown in Eq.~(\ref{Eq2-2}) and Eq.~(\ref{eq5}) by the qubitization. We assume that the system evolves with a time $t=\pi$, then the output state of the QPE circuit is
\begin{equation}
	|1\rangle\left(\sum_{n=1}^{k} \langle u_n|u\rangle_s |u_n\rangle\right)+|0\rangle\left(\frac{|u\rangle_{s}-\sum_{n=1}^{k} \langle u_n|u\rangle_s |u_n\rangle}{\parallel |u\rangle_s-\sum_{n=1}^{k} \langle u_n|u\rangle_s |u_n\rangle\parallel}\right),
\end{equation}
with $|u\rangle_{s}\equiv |a_{k+1}\rangle$. We measure the first register. If the outcome is $0$, we denote the state in the second register as
\begin{equation}
	|u_{k+1}\rangle = \frac{|a_{k+1}\rangle-\sum_{n=1}^{k} \langle u_n|a_{k+1}\rangle |u_n\rangle}{\parallel |a_{k+1}\rangle-\sum_{n=1}^{k} \langle u_n|a_{k+1}\rangle |u_n\rangle\parallel}.
\end{equation}
Thus, a state $|u_{k+1}\rangle$ is constructed and encodes an orthogonal normalized vector $u_{k+1}\in \mathbb{C}^{N}$ such that $a_{k+1}\in span\{u_1,\cdots,u_{k},u_{k+1}\}$ and $u_{k+1}^{\dagger}u_n=0,\forall n\leq k$.
Combining $span\{a_1,\cdots,a_{k}\}=span\{u_1, \cdots,u_{k}\}$ and $u_{n_1}^{\dagger}u_{n_2}=\delta_{n_1n_2}$, $\forall n_1,n_2\leq k$, we have
\begin{equation}
	span\{a_1,\cdots,a_{k+1}\}=span\{u_1, \cdots,u_{k+1}\}
\end{equation}
and
\begin{equation}
	u_{n_1}^{\dagger}u_{n_2}=\delta_{n_1n_2}, \forall n_1,n_2\leq k+1.
\end{equation}
The  quantum circuit  of the $(k+1)$th step of Gram-Schmidt process is  given in Fig.~\ref{fig3}.
\begin{figure}[h]
	\centering
	   	\includegraphics[width=\linewidth]{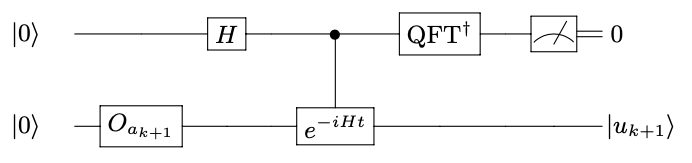}
	\caption{Circuit constructing the $(k+1)$th state $|u_{k+1}\rangle$ based on $\{|u_1\rangle,|u_2\rangle,\cdots,|u_{k}\rangle\}$ and $|a_{k+1}\rangle$. $O_{a_{k+1}}$ is the oracle preparing the quantum state $|a_{k+1}\rangle$, which encodes vector $a_{k+1}$ with $\lceil \log_2 N\rceil$ qubits in the second register. There is only one qubit in the first register. To avoid confusion, the top $H$ is Hadamard gate and the bottom $H$ is Hamiltonian. We post select the first register to be $0$ and readout $|u_{k+1}\rangle$ from the second register. In our algorithm ${\rm QFT}^{\dagger}$ is Hadamard gate because the size of the first register is $1$.}
	\label{fig3}
\end{figure}

 In the linearly dependent case for a set of the $N$-dimensional vectors  $S=\{a_1,a_2,\cdots,a_{M}\}$,  we also first set $|u_{1}\rangle=|a_{1}\rangle$. We assume that $a_{1}, \cdots, a_{k}$ with $k<M$ are linearly independent, then we can apply the algorithm of the linearly independent to these $k$ vectors and obtain $k$  orthogonal normalized vectors  $\{u_1,\cdots,u_{k}\}$ such that $span\{a_1,\cdots, a_{k}\}=span\{u_1,\cdots,u_{k}\}$.  Let us assume that the $(k+1)$th vector $a_{k+1}$ is not linearly independent and can be expressed by the former $k$ vectors, i.e., $a_{k+1}\in span\{a_1,\cdots, a_{k}\}=span\{u_1,\cdots,u_{k}\}$. This can be verified by following the same circuit as in the linearly independent case. That is, we encode the vector $a_{k+1}$ as the quantum state $|a_{k+1}\rangle$ via  Eq.~(\ref{Eq13}) and prepare the input state $|u\rangle_s$ as $|u\rangle_s\equiv |a_{k+1}\rangle$ in the second register. Thus, the output corresponding to the input state $|u\rangle_s=|a_{k+1}\rangle$ is
\begin{equation}
		|1\rangle\left(\sum_{n=1}^{k} \langle u_n|u\rangle_s |u_n\rangle\right)\equiv |1\rangle\left(\sum_{n=1}^{k} \langle u_n|a_{k+1}\rangle |u_n\rangle\right),
\end{equation}
in which the outcome of the measurement on the first register is $1$.  It is impossible that the outcome of the measurement on the first register is 0. Thus, if we run the circuit for a certain time and the outcome for each of the measurements is not $0$, we can infer $a_{k+1}\in span\{a_1,\cdots,a_{k}\}$. As $span\{a_1,\cdots,a_{k+1}\}=span\{u_1,\cdots,u_{k}\}$, we can just move on to the next step of Gram-Schmidt process to find a new $|u_{k+1}\rangle$ by considering the vector $a_{k+2}$. We prove in Appendix~\ref{seca1} that if the circuit is run for $\ln (1/\epsilon)/\epsilon$ times and each of the measurement outcomes is not 0, then with probability larger than $1-\epsilon$, $a_{k+1}$ is linearly dependent of $\{a_1,a_2,\cdots,a_k\}$.

We now summarize our quantum Gram-Schmidt process in Algorithm~\ref{alg3}. By running algorithm~\ref{alg3}, a series of orthogonal normalized bases $\{u_1,u_2,\cdots,u_T\}$ are constructed, satisfying $span\{a_1,a_2,\cdots,a_{M}\}=span\{u_1,u_2,\cdots,u_T\}$ with $T\leq M$, thus the task of orthogonal normalization of vector set is completed. But there are still some questions left. Why do we decide the linearly dependent case with $\ln(1/\epsilon)/\epsilon$ runs? Will errant Hamiltonian simulation lead to loss of orthogonality? How many quantum oracles and gates are needed? We answer these questions in the following section.
\begin{algorithm}[H]
	\caption{Quantum Gram-Schmidt Process for Vector Set Orthogonal Normalization}
	\label{alg3}
	\SetAlgoLined
		\KwIn {$S=\{a_1,a_2,\cdots,a_{M}\}$, $a_i\in \mathbb{C}^N$, error $\epsilon$}
		\KwOut {$S'=\{u_1,u_2,\cdots,u_T\}$ satisfying $u_{t_1}^{\dagger}u_{t_2}=O(\epsilon)$, $\forall t_1\neq t_2$, $span\{a_1,\cdots,a_{M}\}=span\{u_1,\cdots,u_T\}$ with succeeding probability larger than $1-\epsilon$}
		\BlankLine
		$|u_1\rangle\gets |a_1\rangle$\;
		$S'\gets \{u_1\equiv a_1/\parallel a_1\parallel \}$\;
		$H\gets |u_1\rangle\langle u_1|$, $t\gets \pi$, $\epsilon_0\gets \epsilon^4$\;
		\For{$k=1$ \textbf{to} $M-1$}{
		\For{$count=0$ \textbf{to} $1/\epsilon \ln\left(1/\epsilon\right)$}{
		\tcp{count is a counting variable with no other meaning.}
		construct $e^{-iHt}$ to $\epsilon_0$ with qubitization\;
		run circuit with oracle $O_{a_{k+1}}$\;
		measure 1st register to get result $x$\;
		\If{$x=0$}{
		measure 2nd register and get $|\psi\rangle$\;
		readout $|\psi\rangle$ in computational basis as $\psi$\;
		$S'\gets S'\cup\{\psi\}$\;
		$H \gets H+|\psi\rangle\langle \psi|$\;
		\textbf{break}\;
	}
}
}
\textbf{return} $S'$\;
\end{algorithm}

\subsection{Complexity Analysis}\label{4b}

We first provide some lemmas,  then we give our main theorem and prove it based on these lemmas.
\begin{lemma}[Spectral Norm on Kronecker Product]
	\label{lm3}
Suppose $A$ and $B$ are square matrices. Then $\parallel A\otimes B\parallel =\parallel A\parallel  \cdot \parallel B\parallel $.
\end{lemma}
\begin{lemma}[Error of quantum circuit on Errant Hamiltonian Simulation]
	\label{lm4}
Suppose we use the Hamiltonian simulation tools to simulate Hamiltonian $H=\sum_{n=1}^{k}|u_n\rangle\langle u_n|$ for arbitrary time $t$ to error $\epsilon_0$ for each $k=1,2,\cdots,M$
	\begin{equation}
		\parallel e^{-iHt}-U\parallel <\epsilon_0,
	\end{equation}
	then the error of unitary of the whole quantum circuit for constructing $|u_{k+1}\rangle$ as in Fig.~\ref{fig3} can be bounded with
	\begin{equation}
		\label{eq16}
		\parallel U_{\rm real}-U_{\rm exact}\parallel <\epsilon_0.
	\end{equation}
\end{lemma}
\begin{lemma}[Algorithm \ref{alg3} Generate Complete Vector Set]
	\label{lm6}
	Suppose we use Algorithm~\ref{alg3} for Problem~\ref{define1} to generate a series of states successively. Then $span\{a_1,a_2,\cdots,a_{M}\}=span\{u_1,u_2,\cdots,u_T\}$ is satisfied with the probability $\Omega(1)$.
\end{lemma}
\begin{lemma}[Loss of Orthogonality Based on Algorithm \ref{alg3}]
	\label{lm5}
	Suppose we use Algorithm~\ref{alg3} for Problem \ref{define1} to generate  a series of states successively, then the generated states satisfy
	\begin{equation}
		\langle u_{t_1}|u_{t_2}\rangle =O(\epsilon) \;\; \forall t_1\neq t_2.
	\end{equation}
with probability $\Omega(1)$. Therefore,
\begin{equation}
	u_{t_1}^{\dagger}u_{t_2}=O(\epsilon) \;\; \forall t_1\neq t_2
\end{equation}
satisfies with probability $\Omega(1)$
\end{lemma}
Detailed proofs of the lemmas above are given in Appendix \ref{proof1}. Now we give our main result.
\newtheorem{theorem}{Theorem}
\begin{theorem}[Quantum Gram-Schmidt Orthogonalization]
	\label{thm2}
	If we consider a vector set orthogonal normalization problem as defined in Problem \ref{define1}, then there exists a quantum algorithm to generate vector set $S'=\{u_1, u_2,\cdots,u_T\}$ which satisfies $u_{t_1}^{\dagger}u_{t_2}=O(\epsilon)$, $\forall t_1\neq t_2$ and $span\{a_1,a_2,\cdots,a_{M}\}=span\{u_1,u_2,\cdots,u_T\}$, succeeding with the probability larger than $\Omega(1)$. The query complexity of the algorithm is
	\begin{equation}
		O\left(\frac{M^2}{\epsilon}\left(\log_2\frac{1}{\epsilon}\right)^2\right)
	\end{equation}
and $\lceil\log_2 M\rceil+ \lceil\log_2 N \rceil+ 3$ qubits are needed. The total number of additional two-qubit quantum gates is larger than the query complexity by a factor $\log_2 M.$
\end{theorem}
\begin{proof}
	The correctness of Algorithm \ref{alg3} is proved through Lemma \ref{lm6} and \ref{lm5}. Now we calculate the complexity of the algorithm. \\

To simulate LCU Hamiltonian $H=\sum_{n=1}^{k}|u_n\rangle\langle u_n|$ for time $t=\pi$ to accuracy $\epsilon_0=\epsilon^4$ as in Algorithm \ref{alg3}, from Lemma \ref{lm1} we know,
\begin{equation}
O\left(\alpha t+\log_2\left(\frac{1}{\epsilon_0}\right)\right)=O\left(k\pi+4\log_2\left(\frac{1}{\epsilon}\right)\right)
\end{equation}
queries to quantum oracles and $\lceil \log_2 M\rceil+ \lceil \log_2 N\rceil+ 2$ qubits are needed. \\

Thus, The total number of qubits required for the quantum circuit is $\lceil\log_2 M\rceil + \lceil\log_2 N\rceil+ 3$. The number of accesses to quantum oracles required to run the quantum circuit for one time is $O\left(k\pi+4\log_2\left(1/\epsilon\right)\right)$ and the quantum circuit needs to be repeatedly run for at most $\ln (1/\epsilon)/\epsilon$ times for $\forall k$. Therefore, there is an upper bound on the quantum query complexity, i.e., the total number $N_O$ of calls to oracles.
\begin{align}
	N_O&\leq\sum_{k=1}^{M} \frac{1}{\epsilon}\ln\left(\frac{1}{\epsilon}\right) * O\left(k\pi+4\log_2\left(\frac{1}{\epsilon}\right)\right)\nonumber\\
	&=O\left(\frac{M^2}{\epsilon}\left(\log_2\frac{1}{\epsilon}\right)\right)+O\left(\frac{M}{\epsilon}\left(\log_2\frac{1}{\epsilon}\right)^2\right)\nonumber\\
	&=O\left(\frac{M^2}{\epsilon}\left(\log_2\frac{1}{\epsilon}\right)^2\right).
\end{align}
Moreover, there is an upper bound on the total number of additional two-qubit quantum gates $N_G$ with
\begin{align}
	N_G &\leq \sum_{k=1}^{M} \frac{1}{\epsilon}\ln\left(\frac{1}{\epsilon}\right) * O\left(\log_2 k\left(k\pi+4\log_2\left(\frac{1}{\epsilon}\right)\right)\right)\nonumber\\
	&=O\left(\frac{M^2\log_2 M}{\epsilon}\left(\log_2\frac{1}{\epsilon}\right)\right)\nonumber\\
    &+O\left(\frac{M\log_2 M}{\epsilon}\left(\log_2\frac{1}{\epsilon}\right)^2\right)\nonumber\\
	&=O\left(\frac{M^2\log_2 M}{\epsilon}\left(\log_2\frac{1}{\epsilon}\right)^2\right).
\end{align}

Thus, the proof of theorem~\ref{thm2} is completed.
\end{proof}

\subsection{Validations of algorithm}\label{4c}

In this section, we further show the correctness and the performance of our proposed algorithms under different situations. We first calculate the orthogonality of generated vectors $S'=\{u_1,u_2,\cdots,u_{T}\}$ with different dimension of  input vectors $S=\{a_1,a_2,\cdots,a_{M}\}$. Without loss of generality, we assume that the dimension $N$ of the input vector is the same as the number $M$ of input vectors. The loss of orthogonality of the generated $S'=\{u_1,u_2,\cdots,u_{T}\}$ is calculated by
\begin{equation}
	\eta=\parallel S'^{\dagger}S'-I\parallel.
\end{equation}

We randomly generate $N$ vectors $\{a_1,a_2,\cdots,a_{N}\}$ whose dimension is $N$, then line them  into a matrix $\cal{S}\in \mathbb{C}^{N\times N}$ from $a_1$ to $a_N$, in which each column denotes a vector. The conditional number of the matrix $\cal{S}$ is fixed to be $100$. We use classical Gram-Schmidt process and quantum Gram-schmidt process respectively to generate a set of orthogonal normalized vectors $S'=\{u_1,u_2,\cdots,u_{T}\}$. Then the loss of orthogonality of $\{u_1,u_2,\cdots,u_{T}\}$ is calculated. The parameter $\epsilon$ in the quantum Gram-Schmidt process is chosen to be $10^{-4}$. We change the value of $N$ and summarize the results in Fig.~\ref{res6}.

Figure~\ref{res6} shows that the loss of orthogonality of $\{u_1,u_2,\cdots,u_{T}\}$ generated by the quantum Gram-Schmidt process is a little larger than that generated by the classical Gram-Schmidt process. It is because the whole process for the quantum case is simulated on a classical computer, so there is a rounding error in each step of the calculation. The simulation of quantum circuits on a classical computer is inefficient; this requires more computational steps. Therefore, the error of quantum Gram-Schmidt process is a little larger than that of classical Gram-Schmidt process because of the larger rounding error. However, even with a large rounding error, the error of the quantum Gram-Schmidt process is less than $10^{-10}$, which indicates that there is no system error. Thus, the correctness of Algorithm~\ref{alg3} is proved.

We also study the performance of the quantum Gram-Schmidt process when linear independence of $S=\{a_1,a_2,\cdots,a_{N}\}$ is weak. In this case, the conditional number of the matrix $\cal{S}$~$=(a_1,a_2,\cdots,a_{N})$ is large. We randomly generate matrix $\cal{S}$ with different conditional numbers, fixing $\cal{S}$ to be an $8\times 8$ matrix. We use quantum Gram-Schmidt process to generate a series of orthogonal normalized vectors $S'=\{u_1,u_2,\cdots,u_{T}\}$ for $S=\{a_1,a_2,\cdots,a_{N}\}$. The loss of orthogonality is calculated. The parameter $\epsilon$ in our quantum algorithm is chosen to be $10^{-4}$. We change the value of the conditional number and summarize the results in Fig.~\ref{res7}. It is shown from Fig.~\ref{res7} that the loss of orthogonality of $S^{\prime}=\{u_1,u_2,\cdots,u_{T}\}$ remains small even  linear independence of $S=\{a_1,a_2,\cdots,a_{N}\}$ is weak. This shows the robustness of Algorithm~\ref{alg3} in the case that the linear independence of input vectors is weak.

It is noted that we only examine whether Algorithm~\ref{alg3} generates a set of orthogonal normalized vectors, but the completeness of the generated vectors is not examined, i.e.,
\begin{equation}
	span\{a_1,a_2,\cdots,a_{M}\}=span\{u_1,u_2,\cdots,u_{T}\}
\end{equation}
is not proved numerically. This is because we will propose quantum Gram-Schmidt process based QR decomposition algorithm in the next section. When we numerically prove the correctness of our quantum QR decomposition algorithm, the completeness of the generated vectors is verified at the same time.

\begin{figure}[tbph]
	\centering
	\subfigure {\
		\begin{minipage}[b]{\linewidth}
			\centering
			\begin{overpic}[scale=0.5]{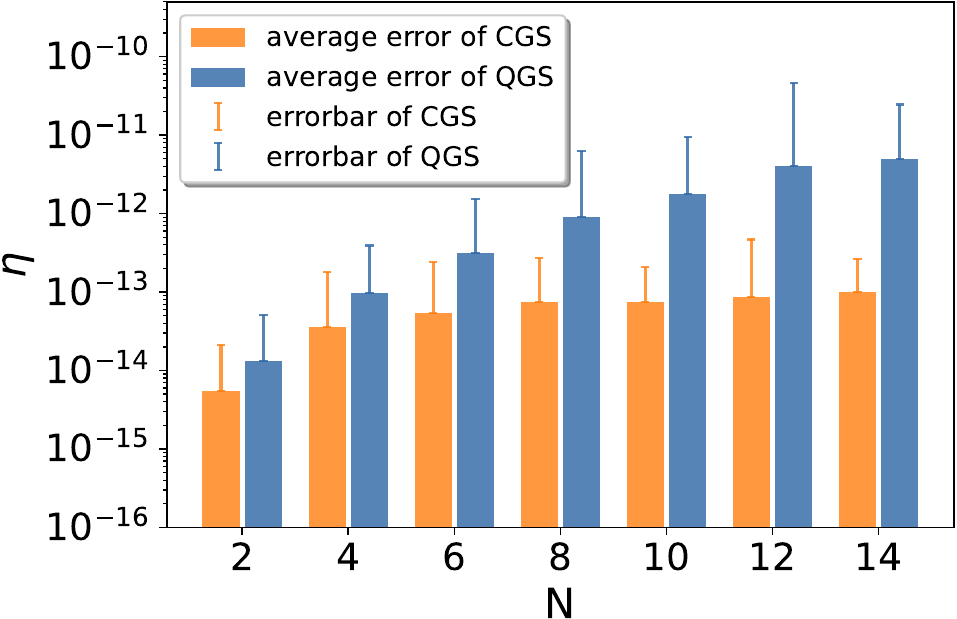}
				\put(92,60){\large\textbf{(a)}}
			\end{overpic}
		\label{res6}
		\end{minipage}
	}
	\subfigure {\
		\begin{minipage}[b]{\linewidth}
			\centering
			\begin{overpic}[scale=0.5]{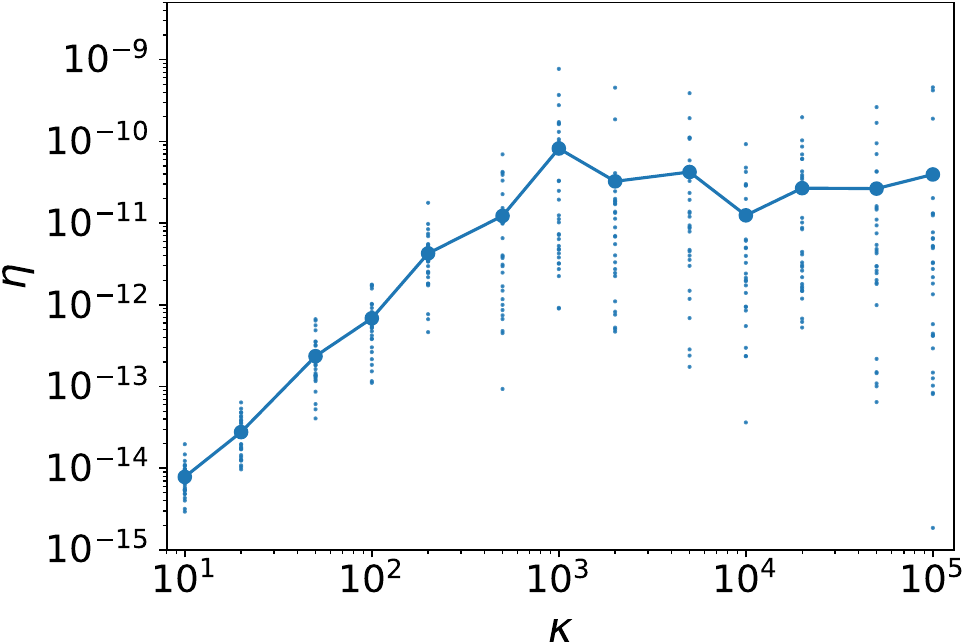}
				\put(91,62){\large\textbf{(b)}}
			\end{overpic}
		\label{res7}
		\end{minipage}
	}
	
	\caption{(a) Loss of orthogonality for different dimensions of the input vector. CGS is for classical Gram-Schmidt process and QGS is for quantum Gram-Schmidt process. $N$ is the system dimension and $\eta$ is the loss of orthogonality. Error bar is also shown. (b) Loss of orthogonality for different conditional numbers of the matrix lined by input vectors. QGS is for quantum Gram-Schmidt process. $\kappa$ is the conditional number of matrix $\cal{S}$ and $\eta$ is the loss of orthogonality. Multiple tests for one conditional number are performed. The result of each test is marked with a small blue dot. The average error is marked with big blue dots.}
\end{figure}

\section{Quantum QR decomposition algorithm}
\label{section5}
In this section, we further propose a quantum QR decomposition algorithm based on our proposed quantum Gram-Schmidt process and a subroutine called quantum inner product estimation. We first introduce a quantum inner product estimation subroutine, and then propose our quantum QR decomposition algorithm.
\subsection{Complexity of Quantum Inner Product Estimation}
\label{sub51}
 Quantum inner product estimation is a method used to calculate the overlap between any two quantum states. Given two $N$-dimension quantum states $|x\rangle$ and $|y\rangle$, a successful quantum inner product estimation algorithm gives the value $|\langle x|y\rangle|^2$ or $\langle x|y\rangle$ with probability $O(1)$. The most common algorithm is the SWAP test. However, SWAP test can only be used to calculate the overlap $|\langle x|y\rangle|^2$, and cannot calculate the value of the inner product $\langle x|y\rangle$ as a complex number. An algorithm for estimating $\langle x|y\rangle$ was proposed~\cite{zhao2021compiling} and can be a subroutine for our quantum QR decomposition algorithm. However, its complexity has not been proven. We now prove the complexity of quantum inner product estimation~\cite{zhao2021compiling} and apply it to the analysis of the complexity of our proposed quantum QR decomposition algorithm.

As shown in Fig.~\ref{fig4}, suppose $O_x$ and $O_y$ are given as oracles for state preparation. $O_x|0\rangle=|x\rangle$, $O_y|0\rangle=|y\rangle$. Then before the quantum measurements, the output of the quantum circuit as in Fig.~\ref{fig4} is
\begin{equation}
	|\phi\rangle=\frac{1}{2}(|0\rangle(|x\rangle+|y\rangle)+|1\rangle(|x\rangle-|y\rangle)).
\end{equation}
Then the probability to obtain $|0\rangle$ is given by
\begin{equation}\label{eq28}
	p=\frac{1}{2}(1+\mathrm{Re}(\langle x|y\rangle))
\end{equation}
 when the first qubit is measured in the computational basis. Thus, the real part of $\langle x|y\rangle$ is encoded in the amplitude of the quantum state $|\phi\rangle$.
\begin{figure}[b]
	\centering
	   	\includegraphics[width=\linewidth]{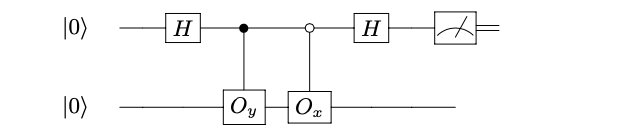}
	\caption{Quantum circuit for estimating Re$\langle x|y\rangle$, the probability that the measurement is 0 is $\left(1+\mathrm{Re}\langle x|y\rangle \right)/2$}
	\label{fig4}
\end{figure}\\

Similarly, by running the quantum circuit as in Fig.~\ref{fig5}, the imaginary part of the quantum state $\langle x|y\rangle$ is encoded into the amplitude. Then the probability to obtain $|0\rangle$  is given by
\begin{equation}\label{eq29}
	p^{\prime}=\frac{1}{2}(1+\mathrm{Im}(\langle x|y\rangle)),
\end{equation}
when  the first qubit is measured in computational basis

\begin{figure}[t]
	\centering
	   	\includegraphics[width=\linewidth]{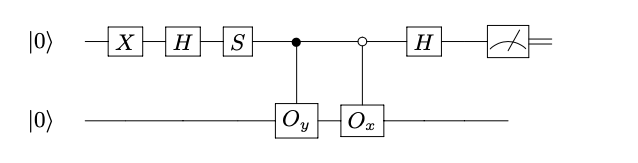}
	\caption{Quantum circuit for estimating Im$\langle x|y\rangle$, the probability that the measurement is 0 is $\left(1+\mathrm{Im}\langle x|y\rangle \right)/2$}
	\label{fig5}
\end{figure}

It is clear that Re$\langle x|y\rangle$ and Im$\langle x|y\rangle$ can be estimated via the measurement on the quantum state $|0\rangle$ for the first qubit when quantum circuits as shown in Fig.~\ref{fig4} and Fig.~\ref{fig5} are separately run for many times. That is, we record the result of each measurement for the first qubit and use the average of the results to obtain the probabilities $p$ and $p^{\prime}$ in Eq.~(\ref{eq28}) and Eq.~(\ref{eq29}), respectively. Thus,  Re$\langle x|y\rangle$ and Im$\langle x|y\rangle$ can be obtained, and the task of estimating the inner product of quantum states is completed.

\begin{lemma}[Quantum inner product estimation]
	\label{lm7}
Given $\epsilon, \, \delta>0$, for any two $N$-dimension quantum states $|x\rangle$ and $|y\rangle$, there exists a quantum algorithm that outputs $\widetilde{\langle x|y\rangle}$ as an estimation of inner product of these two quantum states, satisfying $|\widetilde{\langle x|y\rangle}-\langle x|y\rangle|\leq\epsilon$ with probability larger than $1-\delta$, using
	\begin{equation}
		O\left(\frac{1}{\epsilon^2}\log_2\left(\frac{1}{\delta}\right)\right)
	\end{equation}
	calls to quantum oracles and $\lceil \log_2 n\rceil+1$ qubits, where $n$ is the size of $|x\rangle$ and $|y\rangle$.
\end{lemma}
\begin{proof}
We take the quantum circuits in Fig.~ \ref{fig4} and Fig.~\ref{fig5} for estimating $\mathrm{Re}\langle x|y\rangle$ and $\mathrm{Im}\langle x|y\rangle$, respectively. Suppose the total number of running quantum circuit for estimating $\mathrm{Re}\langle x|y\rangle$ is $N_r$, so the total number of quantum oracle calls is $2N_r$. Each measurement result $X_i=0\;\mathrm{or}\;1,i=1,2,\cdots,N_r$, then according to Hoeffding's inequality,
	\begin{equation}
		P\left(\left|\frac{1}{N_r}\sum_{i=1}^{N_r}X_i-\frac{1+\mathrm{Re}\langle x|y\rangle}{2}\right|\leq\epsilon\right)\geq1-2\exp\left(-2N_r\epsilon^2\right).
	\end{equation}
Thus, when we take $N_r=(16/\epsilon^{2})\log_2\left(4/\delta\right)$ and $\mathrm{Re}\widetilde{\langle x|y\rangle}=(2/N_r)\sum_{i=1}^{N_r}X_i-1$, we have
	\begin{equation}
		\label{eq30}
P\left(|\mathrm{Re}\widetilde{\langle x|y\rangle}-\mathrm{Re}\langle x|y\rangle|\leq\frac{\epsilon}{2}\right)\\
		\geq 1-\frac{\delta}{2}.
	\end{equation}

Similarly, we can estimate $\mathrm{Im}\langle x|y\rangle$ to error $\epsilon/2$ with probability larger than $\delta/2$ with $16\epsilon^{-2}\log_2\left(4/\delta\right)$ runs of circuit in Fig. \ref{fig5}, and obtain
	\begin{equation}
		\label{eq301}
		P\left(|\mathrm{Im}\widetilde{\langle x|y\rangle}-\mathrm{Im}\langle x|y\rangle|\leq\frac{\epsilon}{2}\right)\\
		\geq 1-\frac{\delta}{2}.
	\end{equation}
	\indent Combining Eq.~(\ref{eq30}) and Eq.~(\ref{eq301}), we have
	\begin{equation}
		P\left(|\widetilde{\langle x|y\rangle}-\langle x|y\rangle|\leq \epsilon\right)\geq \left(1-\frac{\delta}{2}\right)^2
		>\; 1-\delta.
	\end{equation}
	The total number of runing quantum circuit is $32\epsilon^{-2}\log_2\left(4/\delta\right)$,  thus the total number of quantum oracle calls is
	\begin{equation}
		64\epsilon^{-2}\log_2\left(4/\delta\right)=O\left(\frac{1}{\epsilon^2}\log_2\left(\frac{1}{\delta}\right)\right)
	\end{equation}
The quantum circuits use $\lceil \log_2 N\rceil$ qubits to encode $|x\rangle$ and $|y\rangle$. One auxiliary qubit is also needed. Thus,  the total number of qubits required is $\lceil \log_2 N\rceil + 1$. The proof of Lemma~\ref{lm7} is completed.
\end{proof}

\subsection{Algorithm Description and Complexity Analysis}
Algorithm~\ref{alg2} shows that the matrices $Q$ and $R$ of the matrix QR decomposition can be constructed through Gram-Schmidt process. The columns of matrix $Q$ are obtained by orthogonal normalization of the columns of the matrix $A$. The non-diagonal matrix elements of matrix $R$ can be obtained by calculating the inner product of the columns of matrix $A$ and matrix $Q$. Based on classical QR decomposition algorithm, we now propose a quantum QR decomposition algorithm. The basic idea of the algorithm is that the matrix $Q$ is constructed by our proposed quantum Gram-Schmidt orthogonalization process, and then the matrix $R$ is calculated with quantum inner product estimation. The proposed algorithm takes a full-rank complex matrix $A\in \mathbb{C}^{N\times M}$ as input where $N\geq M$, and outputs an orthogonal matrix $Q$ and an upper triangular matrix $R$.

Our quantum QR decomposition algorithm is described as follows. We take each column of the matrix $A$ as a $N$-dimensional vector. Thus, the matrix $A$ has $M$ vectors $\{a_1,\cdots, a_{M}\}$. The $N$ components of each  $N$-dimensional vector $a_{m}$ is $a_{nm}$ with $n=1,\cdots, N$. It is clear that $a_{nm}$ is the matrix elements of the matrix $A$. Thus, we use the same steps of the quantum Gram-Schmidt orthogonalization algorithm to obtain orthogonal normalized vectors $\{\tilde{q}_1,\tilde{q}_2,\cdots,\tilde{q}_M\}$, which form the $\tilde{Q}$ matrix with the matrix elements $\tilde{q}_{nm}$. After obtaining the matrix $\tilde{Q}$, we perform quantum inner product estimation algorithm to calculate each element $\tilde{R}_{m_1m_2}=\tilde{q}_{m_2}^{\dagger}a_{m_1}$  of the matrix $\tilde{R}$. Thus, the quantum QR decomposition algorithm is completed.

It is noted that the QR decomposition exists only for matrices with full column rank. Usually, QR decomposition of a full-rank matrix is not unique. Following the classical Gram-Schmidt based QR decomposition algorithm, one solution $A=QR$ is generated. Our quantum algorithm generates matrix $\widetilde{Q}$ and matrix $\widetilde{R}$, which are $\epsilon$ approximation to matrix $Q$ and matrix $R$ generated by classical Gram-Schmidt based QR decomposition algorithm. Our algorithm is summarized in Algorithm \ref{alg4}.
\begin{algorithm}[h]
	\caption{Quantum QR Decomposition}
	\label{alg4}
	\SetAlgoLined
		\KwIn {$A=\{a_1,a_2,\cdots,a_M\}\in \mathbb{C}^{N\times M}$ with full rank and $N\geq M$, error  $\epsilon$}
		\KwOut {$\tilde{Q}=\{\tilde{q}_1,\tilde{q}_2,\cdots,\tilde{q}_M\}\in\mathbb{C}^{N\times M},\tilde{R}\in\mathbb{C}^{M\times M}$, satisfying $\parallel \tilde{q}_m-q_m\parallel =O(\epsilon),\forall m$, $\widetilde{R}_{m_1m_2}=0,\forall m_1>m_2$, and $|\widetilde{R}_{m_1m_2}-R_{m_1m_2}|=O(\epsilon\parallel A\parallel )$, where matrix $Q$ and matrix $R$ is the exact unique decomposition solution via classical Gram-Schmidt process.}
		\BlankLine
		$|\tilde{q}_1\rangle \gets |a_1\rangle$\;
		$H\gets |\tilde{q}_1\rangle\langle \tilde{q}_1|$, $t\gets \pi$, $\epsilon_0\gets \epsilon^4$\;
		\For{$k=1$ \textbf{to} $M-1$}{
		\While{$\rm True$}{
		construct $e^{-iHt}$ to $\epsilon_0$ with qubitization\;
		run circuit with oracle $O_{a_{k+1}}$\;
		measure 1st register to get result $x$\;
		\If{$x=0$}{
		measure 2nd register and get $|q_{k+1}\rangle$\;
		$H \gets H+|\tilde{q}_{k+1}\rangle\langle \tilde{q}_{k+1}|$\;
		\textbf{break}\;
	}
\tcp{If matrix $A$ is not full rank with $a_{k+1}\in span\{a_1,a_2,\cdots,a_k\}$, then this is an endless loop.}
}
}
		$\tilde{Q}\gets\{|\tilde{q}_1\rangle,|\tilde{q}_2\rangle,\cdots,|\tilde{q}_M\}\rangle$ \;
		\For{$m_1=1$ to $M$}{
		\For{$m_2=1$ to $m_1-1$}{
			use QIPE to calculate $\langle\widetilde{q}_{m_2}|a_{m_1}\rangle$ with error rate $\epsilon$, success probability larger than $1-\epsilon/M^2$\;
		$\widetilde{R}_{m_2m_1}\gets \parallel a_{m_1}\parallel \langle\widetilde{q}_{m_2}|a_{m_1}\rangle$\;
	}
		$\widetilde{R}_{m_1m_1}\gets \parallel a_{m_1}-\sum_{m_2=1}^{m_1-1}\widetilde{R}_{m_2m_1}q_{m_2}$$\parallel$ \;
	}
	\textbf{return} $Q$,$R$\;
\end{algorithm}
\begin{theorem}[Quantum QR Decomposition]
	\label{thm3}
Consider a matrix QR decomposition problem as defined in Problem \ref{define2}. Given a full rank matrix $A\in \mathbb{C}^{N\times M}$ with $N\geq M$. Then there exists a quantum algorithm for calculating the QR decomposition of matrix $A$. Suppose the unique exact QR decomposition of matrix $A$ is $A=QR$, then the algorithm generates matrix $\tilde{Q}$ and matrix $\tilde{R}$, satisfying $\parallel\tilde{q_m}-q_m\parallel=O(\epsilon),\forall m$, $R_{m_1m_2}=0,\forall m_1>m_2$, and $|\widetilde{R}_{m_1m_2}-R_{m_1m_2}|=O(\epsilon\parallel A\parallel)$, succeeding with probability $\Omega(1)$. The query complexity of the algorithm is
	\begin{equation}
			O\left(\frac{M^2\log_2 M}{\epsilon}\left(\log_2\frac{1}{\epsilon}\right)^2\right),
	\end{equation}
	and $\lceil \log_2 M\rceil + \lceil \log_2 N\rceil + 3$ qubits are needed. The total number of additional two-qubit quantum gates is larger than the query complexity by a factor $\log_2 m$.
\end{theorem}
\begin{proof}
	We now analyze the complexity of Algorithm~\ref{alg4}, which proves the theorem above. First, Algorithm~\ref{alg4} takes Algorithm~\ref{alg3} as a subroutine. Our quantum QR decomposition algorithm outputs matrix $\tilde{Q}$ as an approximation to matrix $Q$ with Algorithm~\ref{alg3}. From theorem~\ref{thm2} we know that
	\begin{equation}
		O\left(\frac{M^2}{\epsilon}\left(\log_2\frac{1}{\epsilon}\right)^2\right)
	\end{equation}
	calls to quantum oracles, $\lceil \log_2 M\rceil + \lceil \log_2 N\rceil + 3$ qubits are needed, and from Eq.~(\ref{eqa4}) we know
	\begin{equation}
		\parallel \tilde{q}_i-q_i\parallel=O(\epsilon).
	\end{equation}
	
For each nonzero item in matrix $R$, we use quantum inner product estimation algorithm for estimating it. For quantum inner product estimation, we take the estimation accuracy to be $\epsilon$ and the success probability to be $1-\epsilon/{M^2}$. As there are at most $M^2/2$ nonzero items in matrix $R$, so the probability that we estimate each of the item to accuracy $\epsilon$ is larger than
	\begin{equation}
		\prod_{n=1}^{M^2/2}\left(1-\frac{\epsilon}{M^2}\right)>1-\frac{\epsilon}{2}.
	\end{equation}
	\indent So, the probability that each item $\widetilde{R}_{m_1m_2}$ satisfies
	\begin{align}
		&|\widetilde{R}_{m_1m_2}-R_{m_1m_2}|\nonumber\\
		<\;&|\widetilde{R}_{m_1m_2}-\langle \tilde{q}_{m_1}|a_{m_2}\rangle \parallel a_{m_2}\parallel| \nonumber\\
		\indent &+|\langle \tilde{q}_{m_1}|a_{m_2}\rangle \parallel a_{m_2}\parallel-\langle q_{m_1}|a_{m_2}\rangle \parallel a_{m_2}\parallel|\nonumber\\
		<\;&\parallel a_{m_2}\parallel\epsilon + \parallel a_{m_2}\parallel*\parallel |\tilde{q}_{m_1}\rangle-|q_{m_1}\rangle\parallel \nonumber\\
		< \;&2\epsilon \parallel a_{m_2}\parallel\nonumber\\
		=\;& O\left(\epsilon \parallel A\parallel\right)
	\end{align}
	is larger than $1-(\epsilon/2)=\Omega(1)$. For each item estimation, from Lemma \ref{lm7} we know that
	\begin{equation}
		O\left(\frac{1}{\epsilon}\log_2\left(\frac{M^2}{\epsilon}\right)\right)
	\end{equation}
	calls to quantum oracles and $\lceil \log_2 N\rceil$ +1 qubits are needed. So to estimate all the nonzero items, a total of 	
	\begin{equation}
		O\left(\frac{M^2}{\epsilon}\log_2\left(\frac{M^2}{\epsilon}\right)\right)
	\end{equation}
	calls to quantum oracles and $\lceil \log_2 N\rceil$ +1 qubits are needed.

Summarily, the probability that $\parallel\tilde{q_m}-q_m\parallel=O(\epsilon),\forall m$ is $\Omega(1)$ and the probability that $|\widetilde{R}_{m_1m_2}-R_{m_1m_2}|=O(\epsilon\parallel A\parallel), \forall m_1,m_2$ is also $\Omega(1)$. So the probability that the two results hold true at the same time is $\Omega(1)$. In this case, the whole algorithm is successful. The total qubit number needed is
	\begin{equation}
		\max\left(\lceil \log_2 M\rceil+\lceil \log_2 N\rceil+3, \lceil \log_2 N\rceil+1\right),
	\end{equation}
	i.e., $\lceil \log_2 M\rceil + \lceil \log_2 N\rceil + 3$. The total number of calls to quantum oracles is
	\begin{equation}
		\max\left(O\left(\frac{M^2}{\epsilon}\left(\log_2\frac{1}{\epsilon}\right)^2\right),O\left(\frac{M^2}{\epsilon}\log_2\left(\frac{M^2}{\epsilon}\right)\right)\right),
	\end{equation}
	i.e.,
	\begin{equation}
		O\left(\frac{M^2\log_2 M}{\epsilon}\left(\log_2\frac{1}{\epsilon}\right)^2\right).
	\end{equation}
The total number of additional two-qubit quantum gates is larger than the query complexity by a factor $\log_2 M$. Thus, the proof of theorem~\ref{thm3} is completed.
\end{proof}

\subsection{Validations of algorithm}
In this section, we further show the correctness and the performance of our proposed quantum QR decomposition algorithms for the matrix $A$ under different situations. We first calculate the error of the quantum QR decomposition algorithm which outputs an orthogonal matrix $\tilde{Q}$ and an upper triangle matrix $\tilde{R}$. The error of the quantum QR decomposition algorithm is defined as
\begin{equation}
	\eta=\parallel A-\tilde{Q}\tilde{R}\parallel.
\end{equation}
We change the dimension of the input matrix $A$ and calculate the error $\eta$. Without loss of generality, we take matrix $A$ to be a square matrix.

We randomly generate the matrix $A\in\mathbb{C}^{N\times N}$. The conditional number of matrix $A$ is fixed to be 100. We use classical Gram-Schmidt process based QR decomposition algorithm and quantum Gram-schmidt process based QR decomposition algorithm respectively to decompose matrix $A$ into matrix $\tilde{Q}$ and matrix $\tilde{R}$. Then the error $\eta=\parallel A-\tilde{Q}\tilde{R}\parallel$ is calculated. The parameter $\epsilon$ in the quantum Gram-Schmidt process is chosen to be $10^{-4}$. We change the value of $N$ and summarize the results in Fig.~\ref{res8}.

It is shown in Fig.~\ref{res8} that the error of quantum Gram-Schmidt process based QR decomposition algorithm is a little larger than that of the classical Gram-Schmidt process based QR decomposition algorithm. As discussed above in Section~\ref{4c}, it is because the whole process is simulated on a classical computer, so there is a rounding error in each step of calculations, and simulating quantum algorithms requires more computational steps. With rounding error, the error of the quantum Gram-Schmidt process based QR decomposition algorithm is less than $10^{-11}$, which indicates that there is no system error.

It is noted that we have calculated the loss of orthogonality of the matrix $\tilde{Q}$ in Section~\ref{4c}. We show that the loss of  orthogonality of the matrix $\tilde{Q}$, i.e., $\parallel \tilde{Q}^{\dagger}\tilde{Q}-I\parallel$ is less than $10^{-10}$. Meanwhile, the matrix $\tilde{R}$ generated from Algorithm \ref{alg4} is strictly an upper triangle matrix. Thus, we prove numerically that Algorithm~\ref{alg4} generates an orthogonal matrix $\tilde{Q}$ and an upper triangle matrix $\tilde{R}$. Combining the fact that the error $\eta=\parallel A-\tilde{Q}\tilde{R}\parallel$ is less than $10^{-11}$, the correctness of Algorithm~\ref{alg4} is proved.

We also study the performance of the quantum Gram-Schmidt process based QR decomposition algorithm for ill-conditional matrix $A$, of which the conditional number $\kappa$ is large. We randomly generate the matrix $A$ with different conditional numbers, fixing matrix $A$ to be an $8\times 8$ matrix. We use Algorithm~\ref{alg4} to generate matrix $\tilde{Q}$ and matrix $\tilde{R}$ for an input matrix $A$ and calculate the error $\eta$. We change the conditional number $\kappa$ and the parameter $\epsilon$ and summarize the results in Fig.~\ref{res9}. For drawing, we define $\chi=\log_{10} \eta=\log_{10} \parallel A-\tilde{Q}\tilde{R}\parallel$, $\kappa'=\log_{10} \kappa$, and $\epsilon'=\log_{10} \epsilon$.

For a fixed $\epsilon$, it is noticed that as the conditional number grows larger, the error $\eta$ first remains under $10^{-11}$, then grows rapidly and finally becomes stable at the value about $\epsilon$. This is because when the conditional number becomes larger, the linear independence of the columns of the matrix $A$ is weak. Thus,  at certain steps of the quantum Gram-Schmidt process, our algorithm may mistake the linearly independent case as the linearly dependent one. This leads to the rapid growth of the error. It is shown that this error can be reduced by using smaller $\epsilon$. From the numerical simulation, it is shown that when the conditional number $\kappa\in (0,\epsilon^{-1})$, the error is under $10^{-11}$, in this case, the quantum QR decomposition algorithm generates $\tilde{Q}$ which satisfies
\begin{equation}
	span\{\tilde{q_1}, \tilde{q_2},\cdots,\tilde{q_M}\}=span\{a_1,a_2,\cdots,a_M\}.
\end{equation}
However, when $\kappa\in (\epsilon^{-1},+\infty)$, it is shown numerically that the error $\eta$ satisfies
\begin{equation}
	\eta < \epsilon.
\end{equation}
\begin{figure}[tbph]
	\centering
	\subfigure {\
		\begin{minipage}[b]{\linewidth}
			\centering
			\begin{overpic}[scale=0.5]{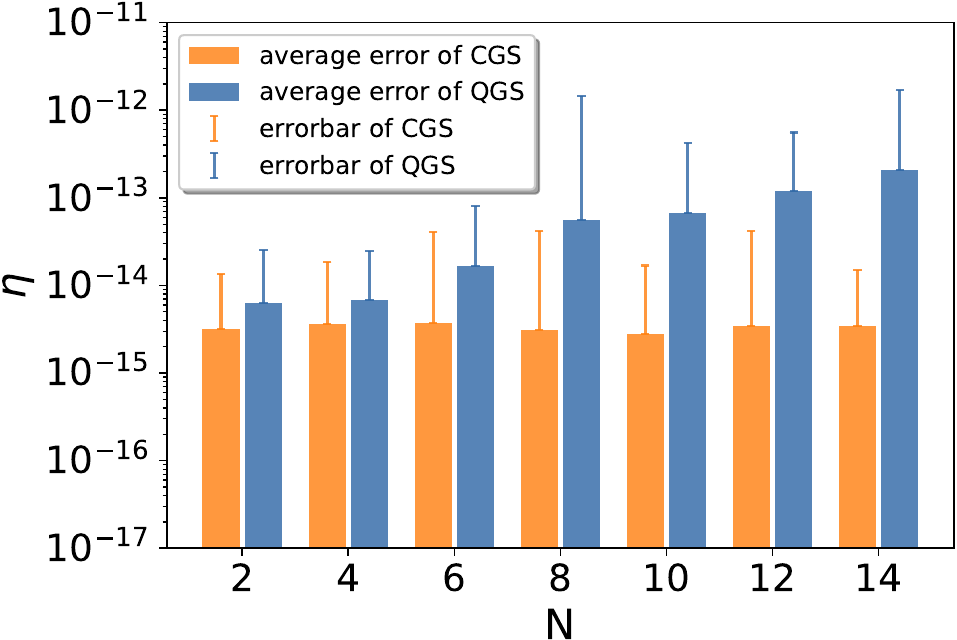}
				\put(92,60){\large\textbf{(a)}}
			\end{overpic}
		\label{res8}
		\end{minipage}
	}
	\subfigure {\
		\begin{minipage}[b]{\linewidth}
			\centering
			\begin{overpic}[scale=0.53]{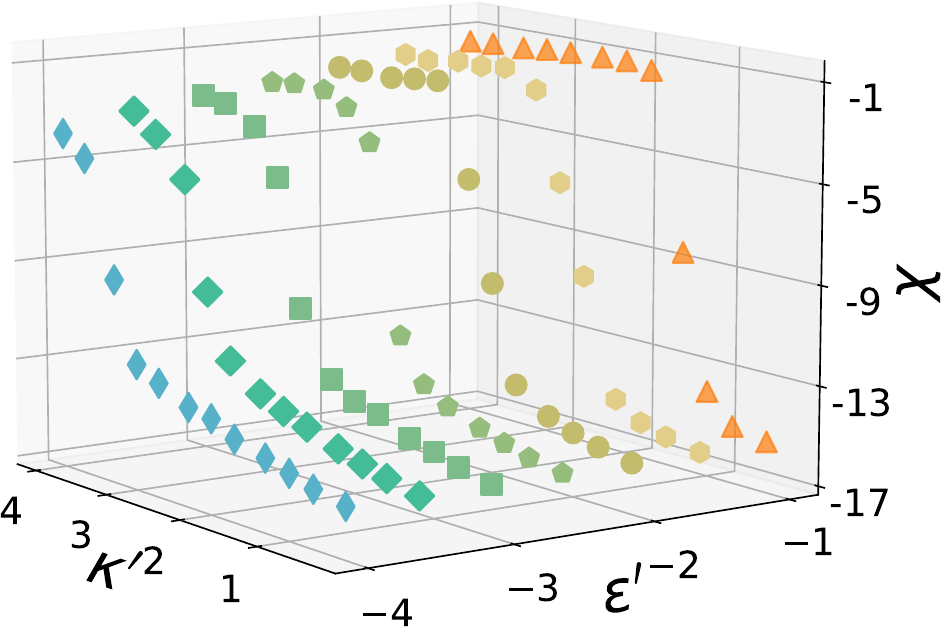}
				\put(90,64){\large\textbf{(b)}}
			\end{overpic}
		\label{res9}
		\end{minipage}
	}
	
	\caption{(a) CGS is for classcial Gram-Schmit process based QR decomposition algorithm and QGS is for quantum Gram-Schmit process based QR decomposition algorithm. $N$ is the system dimension and $\eta$ is the error of QR decomposition. The average error and error bar are calculated and shown. (b) Error of quantum QR decomposition algorithm for different conditional numbers of input matrix and different values of $\epsilon$ in the algorithm. The conditional number is denoted as $\kappa$. We here take $\chi=\log_{10} \eta=\log_{10} \parallel A-\tilde{Q}\tilde{R}\parallel$, $\kappa'=\log_{10} \kappa$, and $\epsilon'=\log_{10} \epsilon$.}
\end{figure}
Therefore, the performance of Algorithm~\ref{alg4} under different size and conditional number of the input matrix is examined. The numerical simulations show the correctness and robustness of Algorithm~\ref{alg4}.
\section{Applications}\label{section6}

The quantum Gram-Schmidt algorithm can be independently applied to solve vector orthogonal normalization problem.  It can also be a subroutine for quantum QR decomposition algorithm.  In this section, we mainly show the applications of the proposed quantum QR  decomposition algorithm by several examples, e.g., linear least squares regression, solving linear equations, and eigenvalues.

\subsection{Linear Least Squares Regression}
The least squares problem is a kind of regression problem, and the general form of least squares problem is
\begin{equation}
	\mathop{\min}_{x}| f\left(x\right)|^2,
\end{equation}
where $f\left(x\right)$ is the residual function and represents the difference between the predicted and measured value, and the loss function is $f\left(x\right)^2$. When $f\left(x\right)=Ax-b$, the least squares problem is a linear least squares problem. The accurate solution can be derived, by letting
\begin{align}
	\frac{\partial \parallel Ax-b\parallel^2}{\partial x}=A^{T}Ax-A^{T}b=0.
\end{align}
Thus $x=\left(A^TA\right)^{-1}A^Tb$. But to obtain $(A^TA)^{-1}$ is hard. However, when matrix $A$ is a full column-rank matrix,  we can use QR decomposition $A=QR$ to simplify the solution, i.e.,
\begin{align}
	x&=\left(A^TA\right)^{-1}A^Tb=\left(R^TQ^TQR\right)^{-1}R^TQ^Tb\nonumber\\
	&=R^{-1}Q^Tb.
\end{align}

As matrix $R$ is an upper triangle matrix, the inversion of matrix $R$ can be obtained classically within $O(M^2)$, where $M$ is the column number of the matrix $A$.  Thus, the whole process is much faster than calculating the inversion of $A^TA$.  Therefore, we can get an algorithm for linear least squares regression using quantum QR decomposition algorithm as a subroutine. We decompose matrix $A$ with quantum QR decomposition algorithm proposed in Algorithm~\ref{alg4}, then solve the linear equations $Rx=Q^Tb$ classically to get $x$.
\begin{figure*}[th]
	\centering
	\subfigure {\
		\begin{minipage}[h]{0.31\linewidth}
			\centering
			\begin{overpic}[scale=0.31]{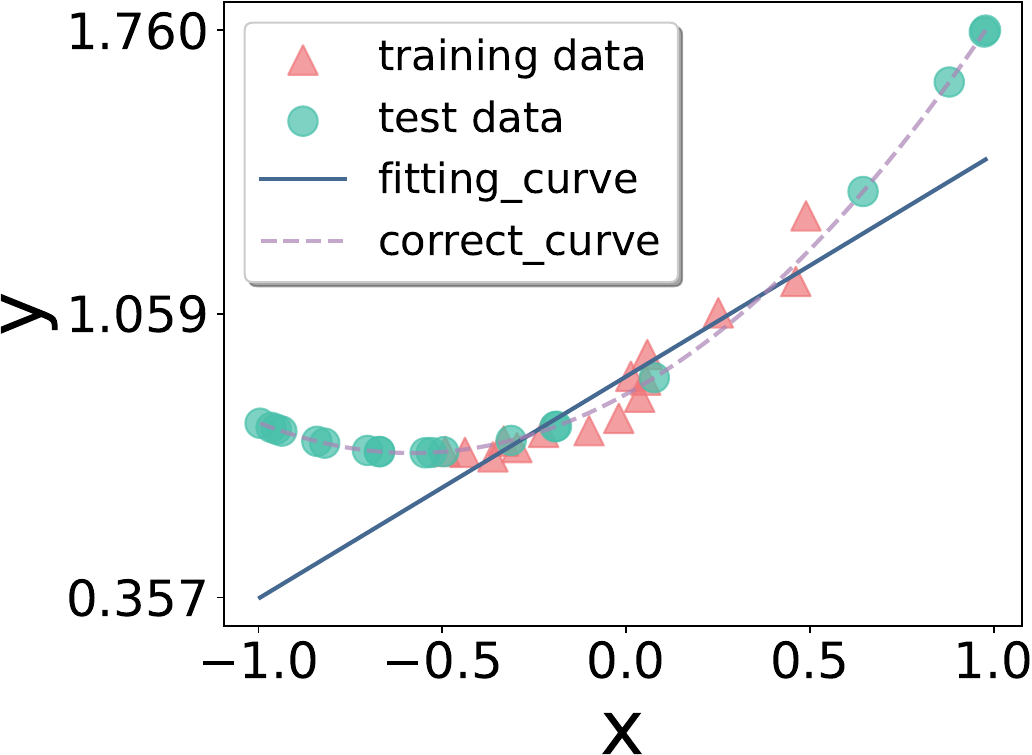}
				\put(87.5,15){\large\textbf{(a)}}
			\end{overpic}
		\end{minipage}
	}
	\subfigure {\
		\begin{minipage}[h]{0.31\linewidth}
			\centering
			\begin{overpic}[scale=0.31]{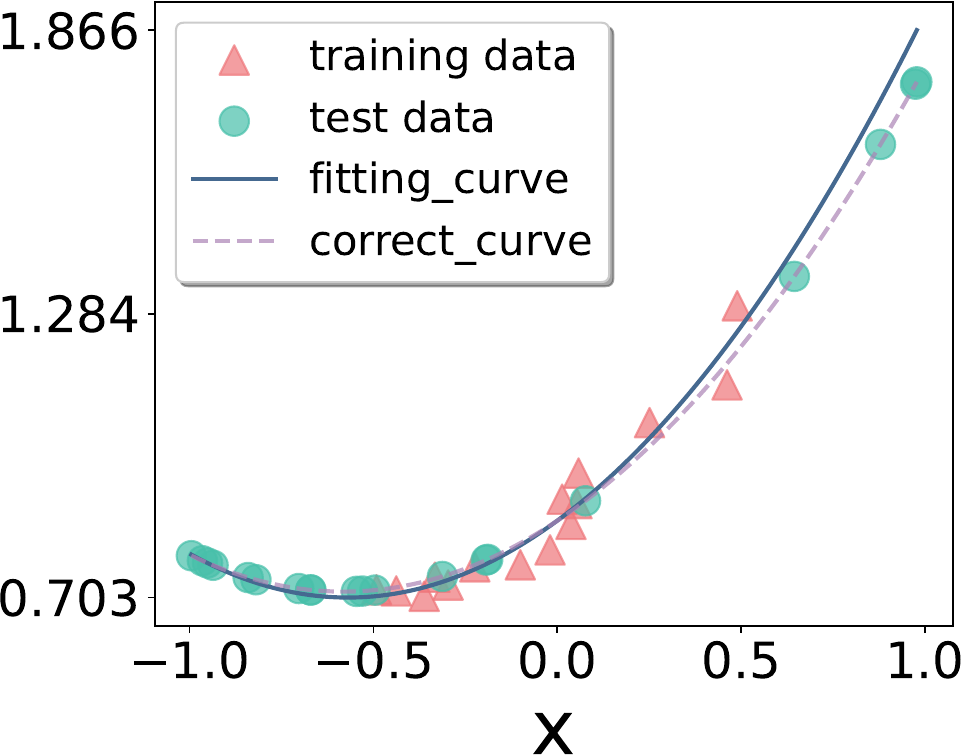}
				\put(86,16){\large\textbf{(b)}}
			\end{overpic}
		\end{minipage}
	}
\subfigure {\
	\begin{minipage}[h]{0.31\linewidth}
		\centering
		\begin{overpic}[scale=0.31]{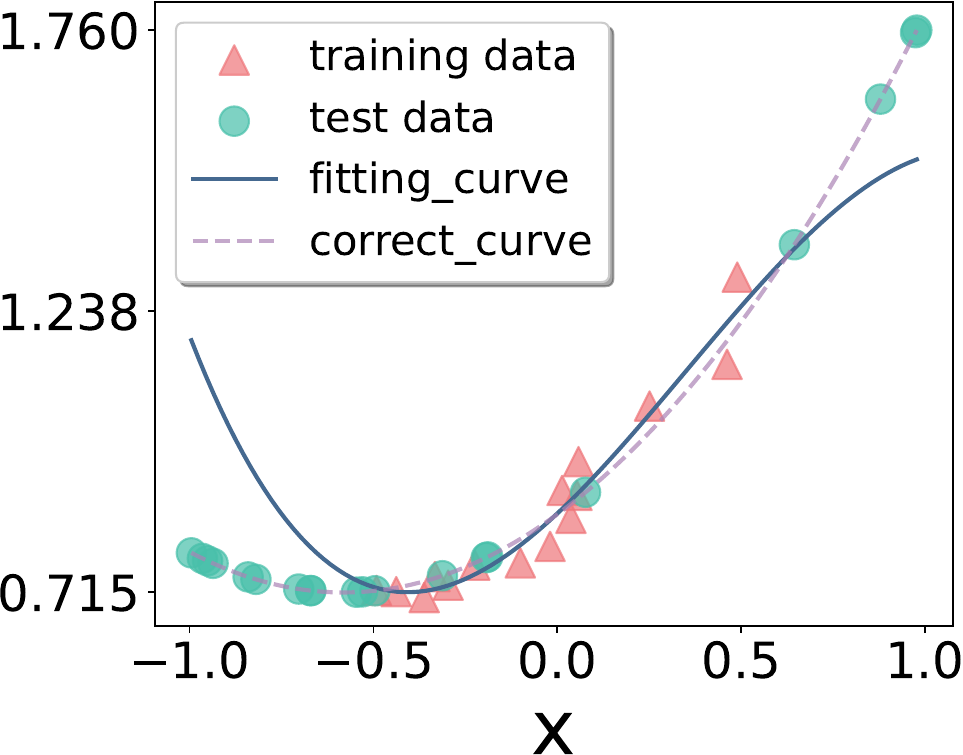}
			\put(87,16){\large\textbf{(c)}}
		\end{overpic}
	\end{minipage}
}
	
	\caption{Fitting $10$ generated data points (the red triangles) with quantum QR decomposition algorithm based least squares fitting method. The order of the polynomial function to be fitted (the purple dashed curve) is chosen to be 2. The order of the fitting polynomial function (the deep blue line) is chosen to be $k$. And test data points (the green dots) are shown. The correct polynomial function is $g(x)=0.43x^2+0.5x+0.86$. (a) $k=1$, under-fitting case. The fitting function is $f(x)=0.54x+0.90$, and the relative error of test data points is 0.249. (b) $k=2$, appropriate fitting case. The fitting function is $f(x)=0.48x^2+0.55x+0.86$, and the relative error of test data points is 0.042.(c) $k=3$, over-fitting case. The fitting function is $f(x)=-0.47x^3+0.50x^2+0.54x+0.86$, and the relative error of test data points is 0.193..}
	\label{res1}
\end{figure*}

Let us now verify our algorithm using an example of polynomial fitting. We assume that a set of data $\{(x_1,y_1), (x_2,y_2), \cdots, (x_n,y_n)\}$ is given. The data set is generated by a polynomial function
\begin{equation}
\label{fit_gen}
  y_i=g(x_i)=\sum_{l=0}^{r} a_l x_i^l.
\end{equation}

Based on this data set, we find a polynomial $f(x)$ to fit the data by least square fitting using quantum QR decomposition as a subroutine so that the mean square error is minimized. We assume
\begin{equation}
\label{fit}
  f(x_i)=\sum_{l=0}^{k} m_l x_i^l.
\end{equation}
Then the coefficients $m_{l}$  can be obtained with least square fitting method with
\begin{align}
	\label{fit_2}
  \overrightarrow{m}&=\min_{\overrightarrow{m}}\sum_{n=1}^{N}\left(\sum_{l=0}^{k} m_l x_i^l-y_i\right)^2\nonumber\\
  &=\min_{\overrightarrow{m}}\parallel X \overrightarrow{m} - Y\parallel ^2,
\end{align}
where
\begin{equation}
\label{lsf}
  X=\begin{pmatrix}
      1 & x_1 & x_1^2 & \cdots & x_1^k \\
      1 & x_2 & x_2^2 & \cdots & x_2^k \\
      1 & x_3 & x_3^2 & \cdots & x_3^k \\
      \vdots & \vdots & \vdots & \ddots & \vdots \\
      1 & x_n & x_n^2 & \cdots & x_n^k \\
    \end{pmatrix}
\end{equation}
and
\begin{equation}
  Y=\begin{pmatrix}
  y_1 & y_2 & y_3 & \cdots & y_n \\
\end{pmatrix}^T
\end{equation}

The items in the polynomial base set $\{1,x,x^2,\cdots,x^k\}$ are orthogonal to each other, the matrix $X$ as in Eq.~(\ref{lsf}) is full rank and has a QR decomposition. Thus, our quantum QR decomposition algorithm can be applied to least square fitting. It is noted that to fit polynomial function $g(x)$ as in Eq.~(\ref{fit_gen}), it is important that the order of polynomial function $k$ as in Eq.~(\ref{fit}) is chosen appropriately. For fixed $r$ as in Eq.~(\ref{fit_gen}), the fitting function $f(x)$  in Eq.~(\ref{fit}) is under fitting when $k<r$,  appropriate when $k=r$, and over fitting when $k>r$.

In Fig.~\ref{res1}, we show the performance of our algorithm. The training dataset is generated with a randomly chosen polynomial function $g(x)=0.86x^2+0.50x+0.43$. We generate ten data points $\{(x_1,y_1),(x_2,y_2),\cdots,(x_{10},y_{10})\}$, which satisfy
\begin{equation}
	y_i=g(x_i)=0.43x_i^2+0.5x_i+0.86,\forall i=1,2,\cdots, 10
\end{equation}
as in Eq.~(\ref{fit_gen}). Then we use a linear function, a quadratic function, and a cubic function to fit the generated training data points respectively, i.e., the value of $k$ as in Eq.~(\ref{fit}) is chosen as $1$, $2$, and $3$, respectively. Quantum QR decomposition algorithm is used to calculate the coefficients in Eq.~(\ref{fit_2}). We plot the generated data points, the correct polynomial function, and the fitted polynomial function in Fig.~\ref{res1} for $k=1,\,2,\,3$, respectively. We also generate fifteen test data points, which satisfy the correct polynomial function $g(x)$ in Eq.~(\ref{fit_gen}). We use test data points to observe if the fitted polynomial function is under-fitting, appropriate, or over-fitting.

The results show in Fig.~\ref{res1} that the linear function cannot fit both the training data and the test data and is the under-fitting case. The quadratic function fits both the training data and the test data well, which is the appropriate case. The cubic function fits well on the training data but cannot fit the test data and is the over-fitting case. The results above agree with the conclusion of classical machine learning. Therefore, the correctness of our quantum QR decomposition algorithm is verified.

Moreover, we use quantum QR decomposition algorithm based least squares fitting method to observe under-fitting, appropriate fitting, and over-fitting phenomena for more general cases. For fixed $r$ as in Eq.~(\ref{fit_gen}), we randomly choose a $r$-order polynomial function $g(x)$ and generate $10$ data points according to $g(x)$. Different orders $k$ as in Eq.~(\ref{fit}) of polynomial functions are chosen for fitting the given data. We calculate the fitting polynomial $f(x)$ using quantum QR decomposition algorithm. We further generate $100$ test data and calculate the fitting error on test data. Then we change $r$ and repeat the process. The relative test error $\epsilon_r$ is recorded and shown in Fig.~\ref{res2}.
	
From classical machine learning theory, we know that the lower left region of the figure is the under-fitting region, satisfying $k<r$. The upper right region of the figure is the over-fitting region, satisfying $k>r$. The test error is large in both of these areas. The polynomial order used for fitting is selected appropriately on the diagonals of the figure satisfying $k=r$, and the test error is minimized. Our quantum QR decomposition algorithm based least squares fitting method reproduces this result as in Fig.~\ref{res2}, which indicates the correctness of our algorithm.

\subsection{Solving Linear Equations}
Solving linear equations is a fundamental problem in linear algebra. Given a matrix $A\in\mathbb{C}^{N\times M}$ and a column vector $b\in\mathbb{C}^N$, define a system of linear equations
\begin{equation}
	Ax=b.
\end{equation}
For the general case, we need to judge whether the system has a unique solution, infinite solutions, or no solution, and give the solution of the system when the system has a unique solution. In general, the matrix $A$ is not necessarily a square matrix. Based on matrix analysis and linear algebra, we know that the system of linear equations has solutions when the given right-side vector $b$ is in the linear space spanned by the columns of matrix $A$. In this case, the system of linear equations has a unique solution when matrix $A$ is invertible, and it has infinite solutions when matrix $A$ is singular. Conversely, if the given vector $b$ is not in the column space of matrix $A$, the system of linear equations has no solution. Based on quantum Gram-Schmidt orthogonalization and quantum QR decomposition, we propose a quantum algorithm to solve linear equations, judge whether the system of linear equations has a unique, infinite, or no solution, and give the solution when the system of linear equations has a unique solution.

\begin{figure}[h]
	\centering
	\includegraphics[width=.75\linewidth]{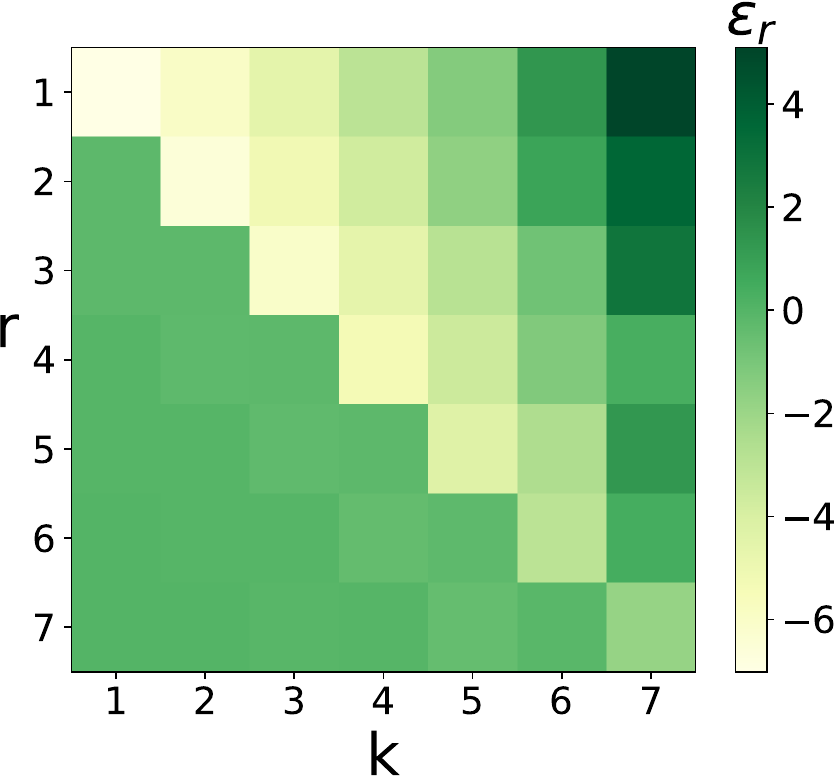}\\
	\caption{The logarithmic relative error $\epsilon_r$of test data is calculated and plotted. $k$ is the order of fitting polynomial function $f(x)$ and $r$ is the order of fitted polynomial function $g(x)$. Underfitting regime is in the lower triangular area. Overfitting regime is in the upper triangular area. Appropriate fitting regime is on the diagonal.  } \label{res2}
\end{figure}

Let us assume that the vector $b$ is in the column space of matrix $A$. Following the quantum circuit used in quantum Gram-Schmidt orthogonalization as in Fig.~\ref{fig3}, we take the state preparation oracle
\begin{equation}
	O_b|0\rangle=|b\rangle,
\end{equation}
where $|b\rangle$ is the amplitude encoding of the vector $b/\parallel b\parallel $  as shown in Eq.~(\ref{Eq13}). We take the Hamiltonian in the circuit as
\begin{equation}
	H=\sum_{m=1}^{M}|a_m\rangle\langle a_m|,
\end{equation}
via the qubitization. Where $|a_m\rangle$ is the amplitude encoding of the vector $a_m/\parallel a_m\parallel $  as shown in Eq.~(\ref{Eq13}), and $a_{m}$ is a vector formed by all matrix elements $a_{nm}$ of the $m$th column of the matrix $A$. When $b\in span\{a_1,a_2,\cdots,a_{M}\}$, i.e
\begin{equation}
	|b\rangle = \sum_{m=1}^{M}\langle a_m|b\rangle |a_m\rangle,
\end{equation}
we have
\begin{align}
	H|b\rangle&=\sum_{m_1=1}^{M}\sum_{m_2=1}^{M}|a_{m_1}\rangle\langle a_{m_2}|a_{m_1}\rangle \langle a_{m_2}|b\rangle\nonumber\\
	&=\sum_{m_1=1}^{M}\sum_{m_2=1}^{M}\delta_{m_1m_2}|a_{m_1}\rangle \langle a_{m_2}|b\rangle=|b\rangle.
\end{align}
Thus, the output of the circuit for the quantum Gram-Schmidt orthogonalization algorithm is
\begin{equation}
	|1\rangle |b\rangle,
\end{equation}
and the probability that we measure $1$ in the first register is $1$. Conversely, if 0 is measured in the first register, then $b\notin span\{a_1,a_2,\cdots,a_{M}\}$, the system of linear equations has no solution. From Eq.~(\ref{eqa3}), we know that if we run the circuit for $(1/\epsilon)\ln\left(1/\epsilon\right)$ times and all the measurement results is $1$, then we can infer that $b\in span\{a_1,a_2,\cdots,a_{M}\}$, as
\begin{equation}
	P\left(1-\sum_{n=1}^{m} |\langle a_i|b\rangle|^2<\epsilon\right)>1-\epsilon.
\end{equation}
\begin{figure*}[th]
	\centering
	\subfigure {\
		\begin{minipage}[h]{0.31\linewidth}
			\centering
			\begin{overpic}[width=5.8cm,height=9cm]{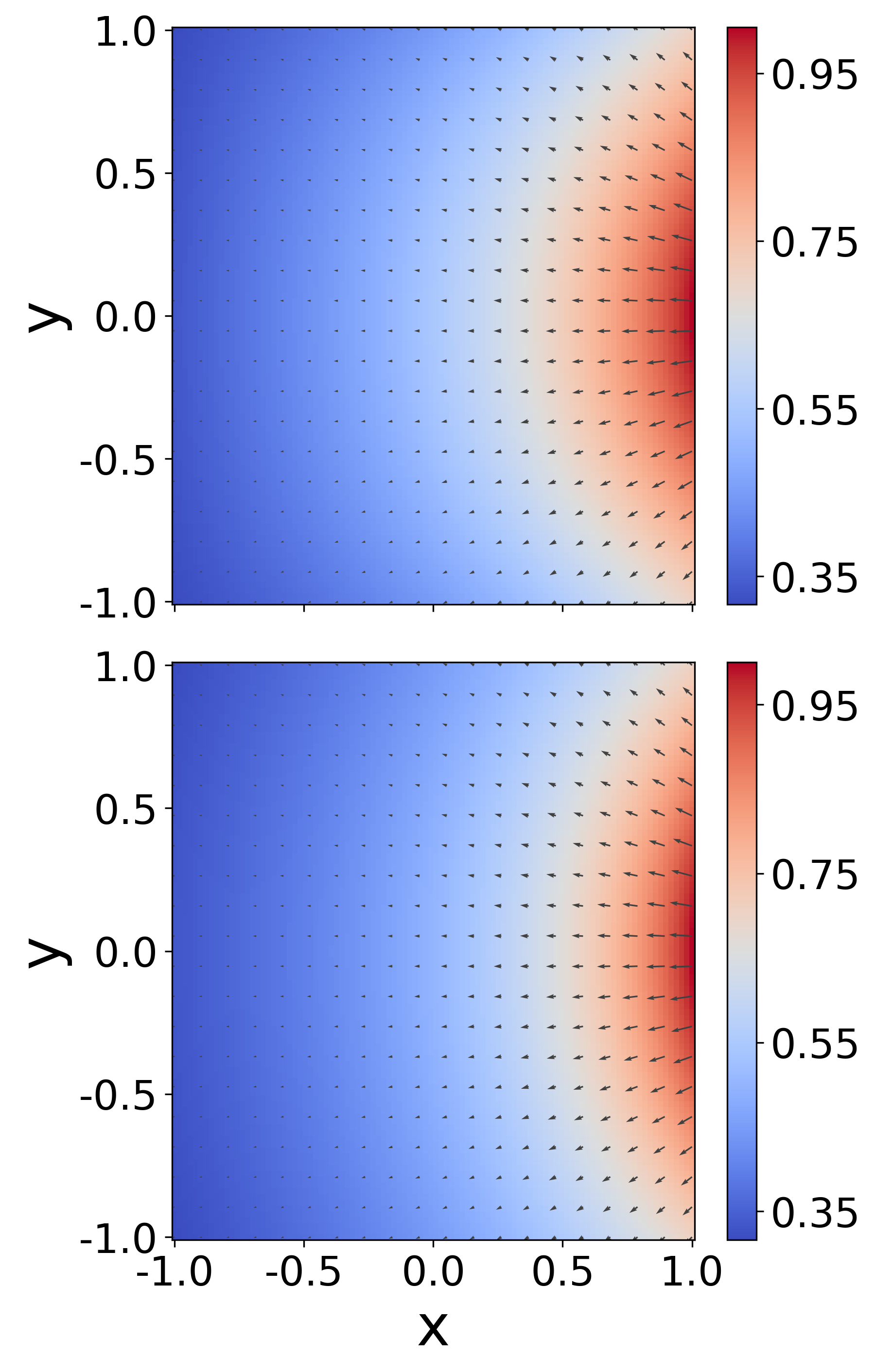}
				\put(30,102){\large\textbf{(a)}}
			\end{overpic}
		\end{minipage}
	}
	\subfigure {\
		\begin{minipage}[h]{0.31\linewidth}
			\centering
			\begin{overpic}[width=5.8cm,height=9cm]{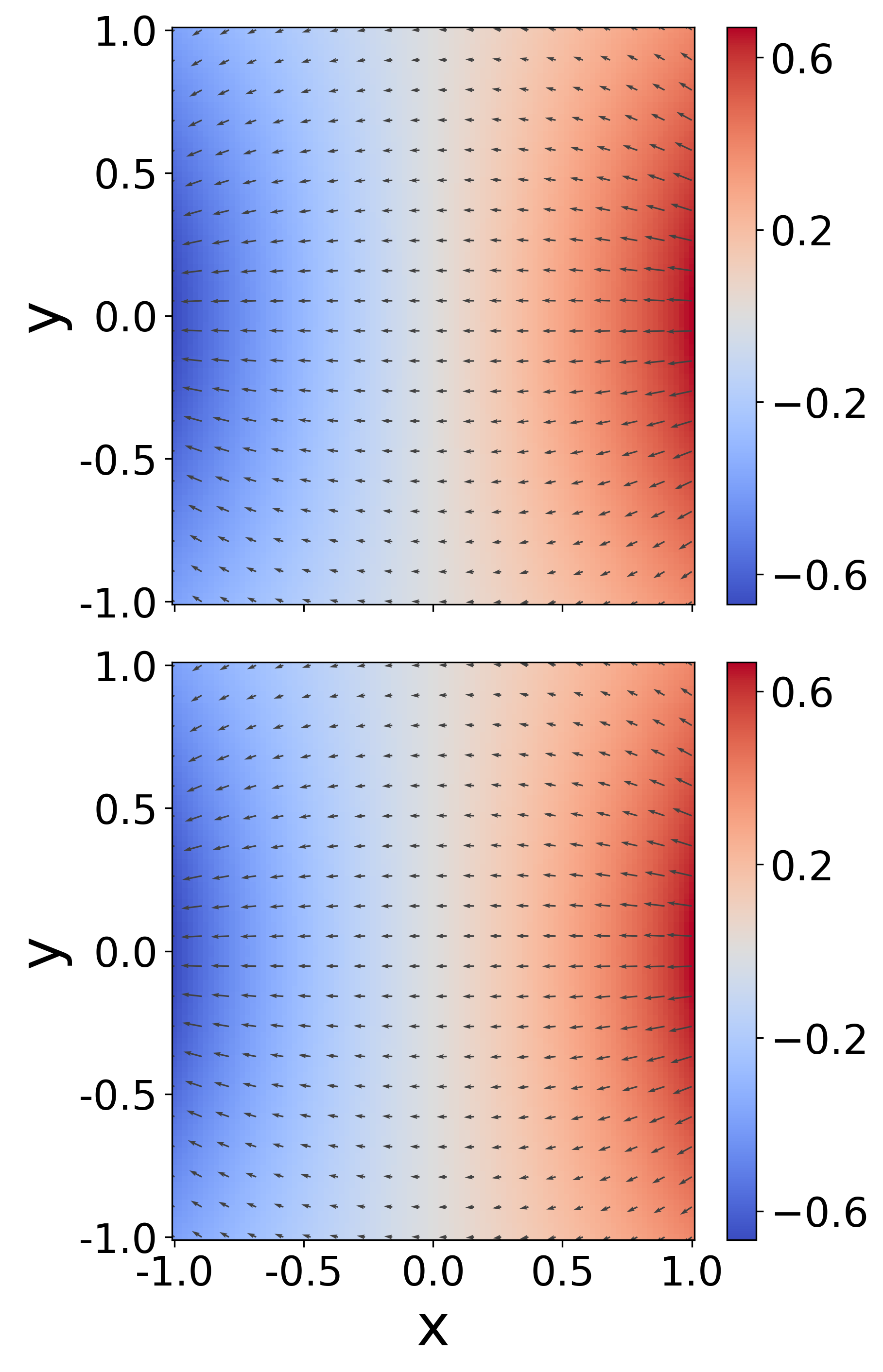}
				\put(30,102){\large\textbf{(b)}}
			\end{overpic}
		\end{minipage}
	}
	\subfigure {\
		\begin{minipage}[h]{0.31\linewidth}
			\centering
			\begin{overpic}[width=5.8cm,height=9cm]{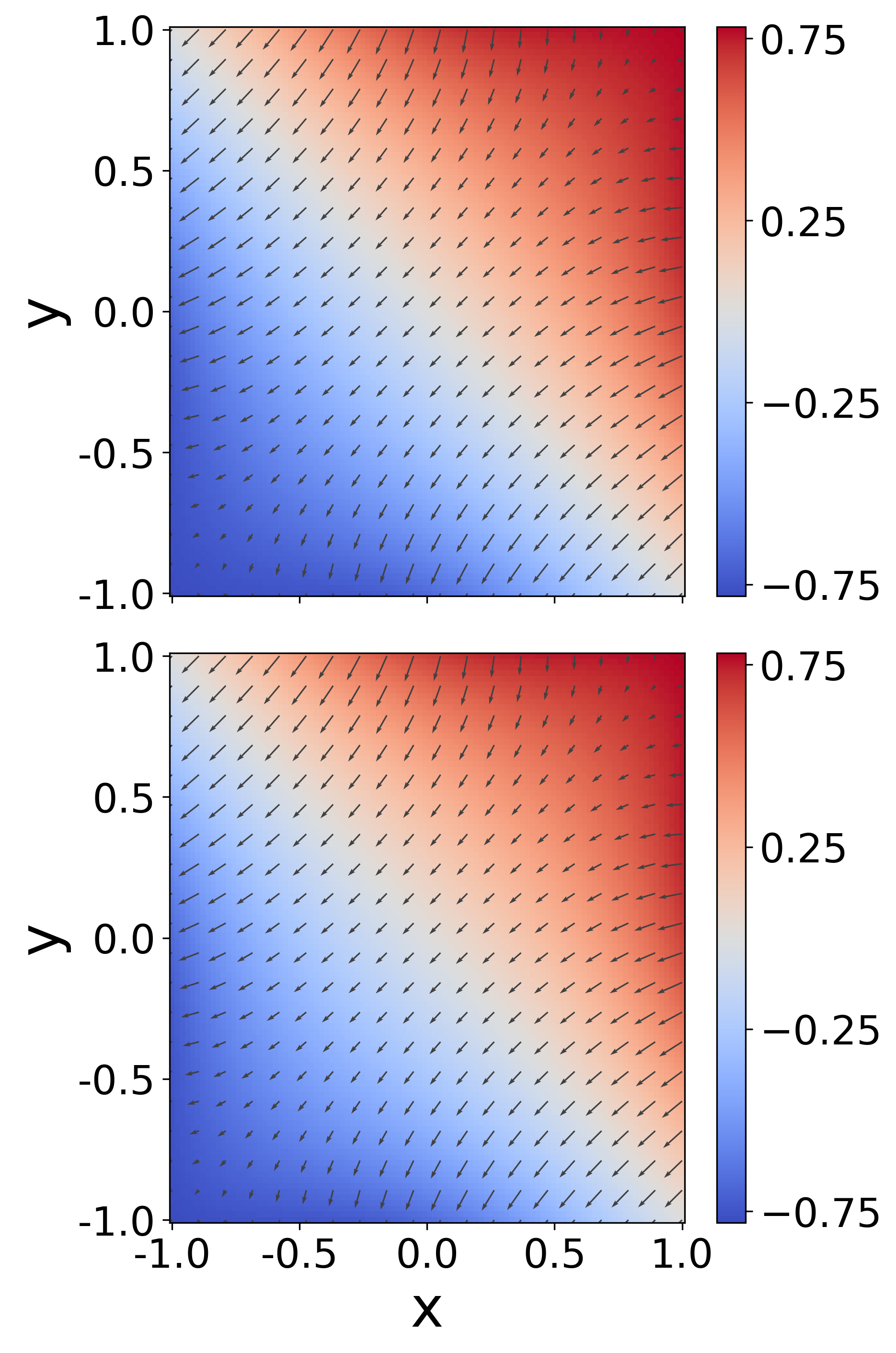}
				\put(30,102){\large\textbf{(c)}}
			\end{overpic}
		\end{minipage}
	}
	
	\caption{The electric potential is plotted as a two-dimensional heatmap. The black arrows represent the electric field intensity, where the directions of arrows represent the directions of electric field intensity, and the length of arrows represent the magnitude of electric field intensity. The top three figures are the solutions from quantum QR decomposition based numerical simulations. The bottom three figures are the exact solutions, which agree with numerical solutions. (a) Electric monopole case, the relative error is 0.053. (b) Electric dipole case, the relative error is 0.057. (c) Electric quadrupole case, the relative error is 0.057.}
	\label{res3}
\end{figure*}

By using the quantum Gram-Schmidt orthogonalization algorithm, we can know if $b$ is or is not in the column space of matrix $A$. If  $b$ is not in the column space of matrix $A$, then linear equations have no solution. If $b$ is in the column space of matrix $A$,  then we perform quantum QR decomposition on the matrix $A$. When the matrix $A$ is not full rank, then linear equations have infinite solutions. In this case, quantum QR decomposition algorithm will go into an endless loop. Conversely, if matrix $A$ is a full-rank matrix, then linear equations have a unique solution. We can obtain the QR decomposition of the matrix $A$; then, the unique solution can be calculated classically with time complexity $O(MN)$ for
\begin{equation}
	x=R^{-1}Q^T b.
\end{equation}

We now further show the correctness of the quantum QR decomposition algorithm for solving the linear equations by an example of solving Laplace equation
\begin{equation}\label{eq63}
  \nabla^2 u=0.
\end{equation}
 At present, the most commonly used method to solve differential equations is the finite difference method. Using methods like the Runge-Kutta method, a differential equation is transformed into difference form, and then transformed into linear equations, which can be solved with quantum QR decomposition algorithm. We verify the performance of our algorithm on a two-dimensional Laplace equation with Dirichlet boundary condition.

 We calculate the electric potential in vacuum. The electric potential $\phi$ satisfies Maxwell's equations
\begin{equation}\label{eq:69}
  -\nabla^2 \phi(x,y)=\frac{\rho(x,y)}{\epsilon_0}.
\end{equation}
We calculate $\phi(x,y)$ on area $\Omega=\{(x,y)|-1<x,y<1\}$, where there is no charge.  That is,  Eq.~(\ref{eq:69})  is changed to a  Laplace equation
\begin{equation}
\nabla^2 \phi(x,y)=0
\end{equation}
for $\forall (x,y)\in \Omega$.  We assume that the electric potential in this area is caused by electric monopole, electric dipole, and electric quadrupole, respectively. In the electric monopole case, we calculate the case that there is a positively charged particle at $(2,0)$. Thus, the boundary condition  and exact solution for equation $\nabla^2 \phi(x,y)=0$ is
\begin{equation}
	\phi(x,y)=\frac{kq}{\sqrt{(x-2)^2+y^2}},
\end{equation}
where $k$ is the Coulomb constant and $q$ is the quantity of electric charge. In the electric dipole case, we calculate the case that there is a positively charged particle at $(2,0)$ and a negatively charged particle at $(-2,0)$. Thus,  the boundary condition and exact solution is
\begin{equation}
	\phi(x,y)=\frac{kq}{\sqrt{(x-2)^2+y^2}}-\frac{kq}{\sqrt{(x+2)^2+y^2}}.
\end{equation}
In the electric quadrupole case, we calculate the case that there are positively charged particles at $(2,0)$ and $(0,2)$, and there are negatively charged particles at $(-2,0)$ and $(0,-2)$. Thus, the boundary condition and exact solution is
\begin{align}
	\phi(x,y)=&\frac{kq}{\sqrt{(x-2)^2+y^2}}-\frac{kq}{\sqrt{(x+2)^2+y^2}}\nonumber\\ &+\frac{kq}{\sqrt{x^2+(y-2)^2}}-\frac{kq}{\sqrt{x^2+(y+2)^2}}.
\end{align}

Without loss of generality, we take $k=q=1$. We use quantum QR decomposition algorithm to solve the electric potential for the monopole case, dipole case, and quadrupole case respectively, and compare with the exact solution. The result is shown in Fig.~\ref{res3}. It is shown that the numerical results agree well with the exact solution, which shows the correctness of our quantum algorithm and its feasibility in solving differential equations.

\subsection{ Solving Eigenvalues}
Our algorithm can also be applied to solve eigenvalues of a full rank matrix by replacing classical algorithm with quantum one.  Classically, the eigenvalues of a full rank matrix  can be solved by QR iteration algorithm as follows. Suppose we have a full rank matrix $A$ and we want to compute its eigenvalues. We first take $A_1=A$, then start to iterate the  matrix. At the $k$-step (starting with $k=1$), we  compute the QR decomposition such that $A_k=Q_kR_k$ with an orthogonal matrix $Q_k$ and an upper triangular matrix $R_k$,  and obtain  $A_{k+1}$ as
\begin{align}\label{eq67}
	A_{k+1}=R_kQ_k.
\end{align}
When $A_k$ converges to the upper triangular matrix, the iteration is completed, and the principal diagonal elements are the eigenvalues of matrix $A$. The classical QR iteration algorithm is almost the most efficient method to solve all eigenvalues of matrices up to now. However, we find that this classical iteration QR algorithm can be replaced by the quantum one, in which classical QR decomposition in each iteration is replaced by the quantum one. We have proved that the computational complexity of the quantum QR decomposition algorithm is less than the classical one. Thus, the computational complexity of the quantum QR iteration algorithm is also less than the classical one.

To more concretely show the application of our algorithm to the problem, we now take examples of solving energy spectra of quantum Ising and Heisenberg models, which have extensively been studied in quantum many-body physics and condensed matter physics~\cite{schollwock2005density,schollwock2011density,orus2008infinite,orus2014practical,fehske2007computational}. We will show that our quantum algorithm can be applied to solve these problems, which are not easy to be solved in the large scale systems by classical numerical one. We know that $N$-site open boundary one dimensional Ising model has the Hamiltonian
\begin{equation}
  H=-h\sum_{n=1}^{N}\sigma_x^i-J\sum_{n=1}^{N-1}\sigma_z^i\sigma_z^{i+1}.
\end{equation}
However, $N$-site open boundary one dimensional Heisenberg model has the Hamiltonian
\begin{equation}
  H=-J\sum_{n=1}^{N-1}\left(\sigma_x^i\sigma_x^{i+1}+\sigma_y^i\sigma_y^{i+1}+\sigma_z^i\sigma_z^{i+1}\right).
\end{equation}
We calculate the energy spectra of one dimensional Ising model and Heisenberg model taking the size to be $5$ and all parameters to be $1$. Using quantum QR decomposition algorithm, we perform iterative QR algorithm to obtain the eigenvalues of one dimensional Ising model and Heisenberg model.

\begin{figure}[tbph]
	\centering
	\subfigure {\
		\begin{minipage}[b]{\linewidth}
			\centering
			\begin{overpic}[scale=0.55]{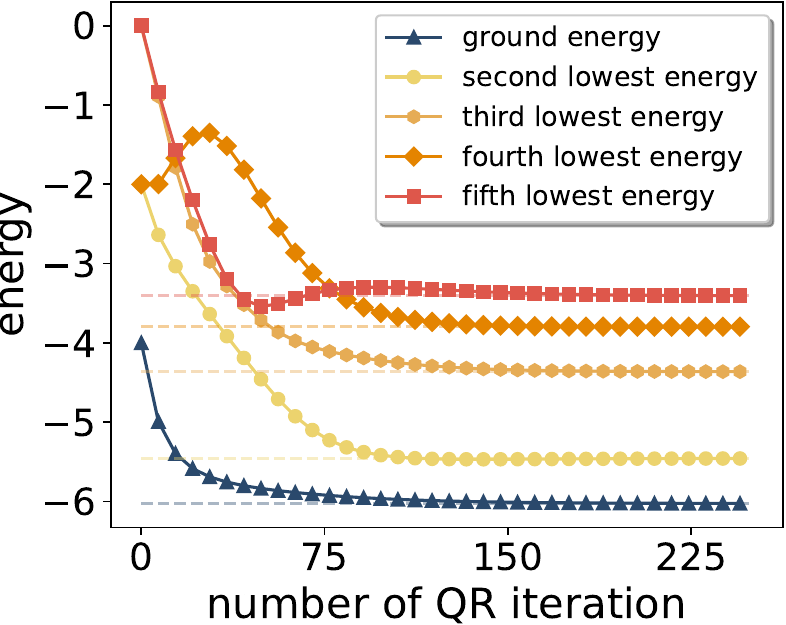}
				\put(-5,75){\large\textbf{(a)}}
			\end{overpic}
		\label{res4}
		\end{minipage}
	}
	\subfigure {\
		\begin{minipage}[b]{\linewidth}
			\centering
			\begin{overpic}[scale=0.55]{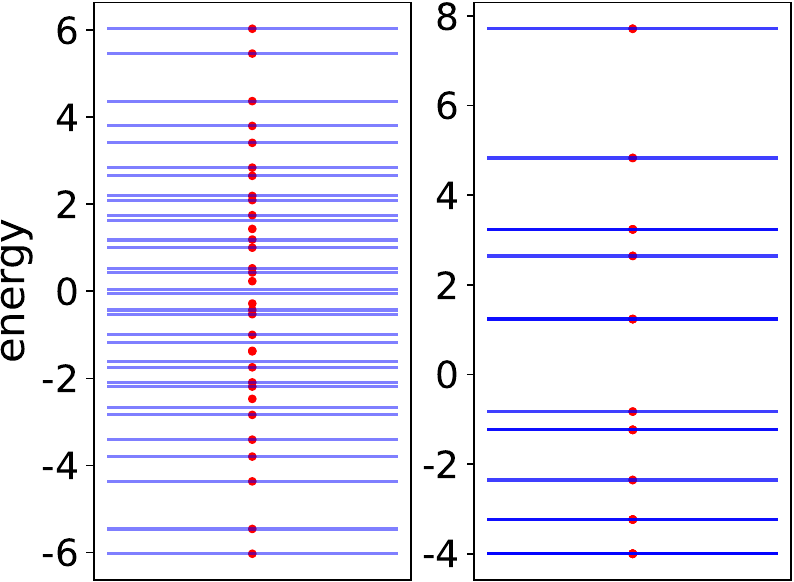}
				\put(-5,70){\large\textbf{(b)}}
			\end{overpic}
		\label{res}
		\end{minipage}
	}
	
	\caption{(a) Convergence of few lowest eigenvalues of Ising model calculated by using quantum QR decomposition based QR algorithm. (b) Comparison of energy spectra of Ising model and Heisenberg model calculated by using quantum QR decomposition based QR algorithm with exact solutions of energy spectra.}
\end{figure}

We first study the convergence of calculated eigenvalues of Ising model. We take the diagonal elements of the matrix $A_k$ obtained by the iterative QR algorithm at each step, take the lowest values of the diagonal elements as the calculated values of the lowest energy levels, and study how these values change with the increase of the number of iterative steps. The convergence of a few lowest energy levels through the iteration of QR process is shown in Fig.~\ref{res4}, where the dashed lines represent the exact values of Ising model energy spectra calculated by exact diagonalization method. It is shown that the calculated values of the few lowest energy levels converge to the exact values accurately. For an Ising model with the size $N$, the dimension of the Hamiltonian for Ising model is $2^N$. It is shown in Fig.~\ref{res4} that QR algorithm converges fast with $O(2^N)$ iterations.

Then we compare the calculated energy spectra of Ising model and Heisenberg model with exact energy spectra, which are calculated by exact diagonalization method. The results are shown in Fig.~\ref{res}. The blue lines are the exact energy spectra of Ising model (left) and Heisenberg model (right), and the red dots are the energy spectra calculated by using quantum QR decomposition based QR algorithm. It is shown that the red dots agree with the blue lines for both Ising model and Heisenberg model, showing that our quantum QR decomposition based QR algorithm can successfully calculate the energy spectra of both these two models accurately. Thus, our algorithm can be applied to eigenvalue finding with high accuracy.

\section{Discussions and Conclusions}
\label{section7}

In summary, we propose quantum algorithms for vector set orthogonal normalization and matrix QR decomposition respectively, which are basic tasks in matrix analysis and linear algebra with various applications in many fields.  The proposed algorithms achieve polynomial speedup over the best classical algorithms and quantum algorithms, scaling $O(N^2)$ in the $N$-dimensional system. Our algorithms scaling $O(1/\epsilon^2)$ in the tolerant error $\epsilon$  may be further improved. We use repetitive measurements to estimate the inner product of two quantum states, leading to this $O(1/\epsilon^2)$  scaling.  Quantum amplitude estimation methods can be applied to estimate quantum inner product~\cite{brassard2002quantum, giurgica2022low, grinko2021iterative}, which may improve the scaling in tolerant error $\epsilon$ to $O(1/\epsilon)$.

Despite of these shortcomings, the scaling $O(N^2)$ for the $N$-dimensional system is very promising. Meanwhile, we study several applications of our quantum algorithms to linear least squares regression, solving linear equations and eigenvalues of the matrix. The correctness and efficiency of our algorithms are also proved with detailed proof and numerical simulations. We believe that our algorithms have broad applications for serving as subroutines to solve many important problems. Our quantum algorithms for some applications can be demonstrated by NISQ-era superconducting quantum computing chips.

\begin{acknowledgments}
	This work was supported by Innovation Program for Quantum Science and Technology (Grant No. 2021ZD0300201).
\end{acknowledgments}
\appendix
\section{Derivation of Quantum Phase Estimation}
\label{ap1}
\indent The quantum circuit for quantum phase estimation is given in Fig. \ref{fig6} when we limit the number of qubits in the first register to be 1. Suppose the Hamiltonian $H$ has a spectral decomposition
\begin{equation}
	H=\sum_{n=1}^{N}\lambda_n|u_n\rangle\langle u_n|,
\end{equation}
where $N$ is the dimension of $H$ and each $\lambda_n$ can be 0. Then $\{|u_1\rangle,|u_2\rangle,\cdots,|u_N\rangle\}$ is a complete set, so the input state of the second register $|u\rangle_s$ can be written as
\begin{equation}
	|u\rangle_s = \sum_{n=1}^{N} \langle u_n|u\rangle_s |u_n\rangle
\end{equation}
So the input state of the circuit is
\begin{equation}
	|0\rangle_f|u\rangle_s = \sum_{n=1}^{N} |0\rangle_f\langle u_n|u\rangle_s  |u_n\rangle.
\end{equation}
\begin{figure}[b]
	\centering
	\includegraphics[width=.95\linewidth]{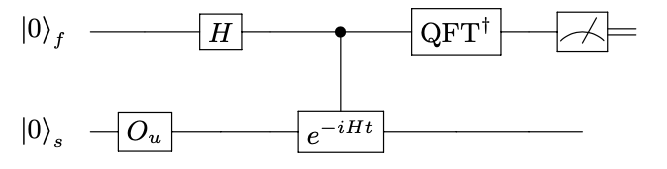}
	\caption{Quantum circuit for quantum phase estimation when the size of first register is 1. To avoid confusion, the top $H$ represents Hadamard gate and the bottom $H$ is Hamiltonian.}
	\label{fig6}
\end{figure}
After first Hadamard gate, the state of system is
\begin{equation}
	\sum_{n=1}^{N} \frac{|0\rangle+|1\rangle}{\sqrt{2}}\langle u_n|u\rangle_s  |u_n\rangle.
\end{equation}
As
\begin{equation}
	e^{-iHt}|u_n\rangle = e^{-i\lambda_nt}|u_n\rangle,
\end{equation}
so after the controlled-$e^{-iHt}$ gate, the state of system is
\begin{equation}
	\sum_{n=1}^{N} \langle u_n|u\rangle_s \frac{|0\rangle_f+e^{-i\lambda_nt}|1\rangle_f}{\sqrt{2}}|u_n\rangle .
\end{equation}
After the final Hadamard gate, the state of the system is
\begin{equation}
	\sum_{n=1}^{N} \langle u_n|u\rangle_s \frac{\left(|0\rangle_f+|1\rangle_f\right)+e^{-i\lambda_nt}\left(|0\rangle_f-|1\rangle_f\right)}{2}|u_n\rangle .
\end{equation}
Rewrite the equation.
\begin{equation}
	\sum_{n=1}^{N} \langle u_n|u\rangle_s \left(\frac{1+e^{-i\lambda_n t}}{2}|0\rangle_f+\frac{1-e^{-i\lambda_n t}}{2}|1\rangle_f\right)|u_n\rangle.
\end{equation}

If the eigenvalues of Hamiltonian $H$ satisfy
\begin{align}
	\lambda_n&=1,\;\;n=1,2,\cdots,k\nonumber\\
	\lambda_n&=0,\;\ n=k+1,\cdots,N
\end{align}
as in the quantum Gram-Schmidt process and the evolution time $t$ of the Hamiltonian $H$ is
\begin{equation}
	t=\pi,
\end{equation}
The final state of the quantum circuit then is
\begin{align}
	&\sum_{n=1}^{k} \langle u_n|u\rangle_s \left(\frac{1+e^{-i\pi}}{2}|0\rangle_f+\frac{1-e^{-i\pi}}{2}|1\rangle_f\right)|u_n\rangle\nonumber\\
	+&\sum_{n=k+1}^{N} \langle u_n|u\rangle_s \left(\frac{1+e^{-i\pi*0}}{2}|0\rangle_f+\frac{1-e^{-i\pi*0}}{2}|1\rangle_f\right)|u_n\rangle,
\end{align}
i.e,
\begin{equation}
	|1\rangle_f\left(\sum_{n=1}^{k} \langle u_n|u\rangle_s |u_n\rangle\right)+|0\rangle_f\left(\sum_{n=k+1}^{N} \langle u_n|u\rangle_s |u_n\rangle\right).
\end{equation}
So the probability that measured result from the first register is 0 is
\begin{equation}
	1-\sum_{n=1}^{k}|\langle u_n|b\rangle|^2.
\end{equation}
When we measure 0 from the first register, then the state of the second register is then
\begin{equation}
	|\psi\rangle=\frac{\sum_{n=k+1}^{N} \langle u_n|b\rangle |u_n\rangle}{\parallel \sum_{n=k+1}^{N} \langle u_n|b\rangle |u_n\rangle\parallel },
\end{equation}
which satisfies
\begin{equation}
	|\langle \psi|u_n\rangle|^2=0,\forall n=1,2,\cdots, k.
\end{equation}
.
\section{Proof of Lemmas in the section \ref{4b}}
\label{proof1}
\subsection{Proof of Lemma \ref{lm3}}
\begin{proof}
	\begin{align}
		\parallel A\otimes B\parallel_2^2 &= \left(\sqrt{\lambda_{max}\left(A\otimes B\right)^{\dagger}\left(A\otimes B\right)}\right)^2\nonumber\\
		&= \lambda_{max}\left(A\otimes B\right)^{\dagger}\left(A\otimes B\right)\nonumber\\
		&= \lambda_{max}\left(A^{\dagger}A \otimes B^{\dagger}B\right)\nonumber\\
		&= \lambda_{max}\left(A^{\dagger}A\right) \cdot\lambda_{max}\left(A^{\dagger}A\right)\nonumber\\
		&= \parallel A\parallel_2^2 \cdot \parallel B\parallel_2^2
	\end{align}
	\indent Therefore $\parallel A\otimes B\parallel_2=\parallel A\parallel_2 \cdot \parallel B\parallel _2$.
\end{proof}
\subsection{Proof of Lemma \ref{lm4}}
\begin{proof}
	Suppose the qubit number of the first register is 1. Then the exact and errant QPE unitary is
	\begin{equation}
		U_{exact}=\left(QFT^{\dagger}\otimes I\right)\left(C-U\right)\left(H \otimes I\right)
	\end{equation}
	\begin{equation}
		U_{real}=\left(QFT^{\dagger}\otimes I\right)\left(C-\exp(-iHt)\right)\left(H \otimes I\right)
	\end{equation}
	with $\parallel U-\exp(-iHt)\parallel <\epsilon_0$.
	
Define $\Delta U=\left(C-U\right)-\left(C-\exp(-iHt)\right)$. From Lemma \ref{lm3}, we get the following equation.
	\begin{align}
		\parallel \Delta U\parallel &=\parallel \left(C-U\right)-\left(C-\exp(-iHt)\right)\parallel \nonumber\\
		&= \parallel |1\rangle\langle 1|\otimes\left(U-\exp(-iHt)\right)\parallel \nonumber\\
		&=\parallel 1\rangle\langle 1 \parallel \parallel U-\exp(-iHt)\parallel \nonumber\\
		& <\epsilon_0.
	\end{align}
Therefore,
	\begin{align}
		& \parallel U_{real}-U_{exact}\parallel  \nonumber \\
		\leq & \;\parallel \left(QFT^{\dagger}\otimes I\right)\parallel \parallel \left(H\otimes I\right)\parallel \parallel \Delta U\parallel \nonumber\\
		=& \parallel \Delta U\parallel\nonumber\\
		<&\; \epsilon_0
	\end{align}.
	
	Therefore, we prove Lemma \ref{lm4}.
\end{proof}
\subsection{Proof of Lemma \ref{lm6}}
\label{seca1}
\indent In Algorithm \ref{alg3}, we generate a series states successively, so we use mathematical induction to prove Lemma \ref{lm6}.
\begin{proof}
	\indent Firstly, as we take $|u_1\rangle\equiv|a_1\rangle$, So $span\{a_1\}=span\{u_1\}$. It is clear that $|u_t\rangle$ and $|a_m\rangle$ is the amplitude encoding of $N-$dimension vector $u_t$ and $a_m$.
	\begin{equation}
		|a_m\rangle=\sum_{n=1}^{N}\frac{a_{nm}}{\parallel a_m \parallel} |n-1\rangle.
	\end{equation}
	\begin{equation}
		|u_t\rangle=\sum_{n=1}^{N}\frac{u_{nt}}{\parallel u_t \parallel} |n-1\rangle.
	\end{equation}
	
	Suppose after the first $k$ steps of quantum Gram-Schmidt process, we have generated $\{|u_1\rangle,|u_2\rangle,\cdots,|u_t\rangle\}$ for $\{a_1,a_2,\cdots,a_k\}$, satisfying $span\{a_1,a_2,\cdots,a_k\}=span\{u_1,u_2,\cdots,u_t\}$ with probability $\Omega(1)$. It is clear that $t$ is not always equal to $k$. It is because $\{a_1,a_2,\cdots,a_k\}$ is not always a set of linearly independent vectors. Therefore, $k\geq t$.
	
	In the $(k+1)$th step of the quantum Gram-Schmidt process, the input state of the second register is $|a_{k+1}\rangle$ which can be written as
	\begin{equation}
		\label{eqa1}
		|a_{k+1}\rangle=\sum_{n=1}^{t}\langle u_n|a_{k+1}\rangle |u_n\rangle + c |\psi\rangle,
	\end{equation}
	where we denote $c$ to be a real number
	\begin{equation}
		c=\sqrt{1-\sum_{n=1}^{t}\parallel\langle u_n|a_{k+1}\rangle\parallel^2}.
	\end{equation}
	After running the quantum circuit of the $(k+1)$th step of quantum Gram-Schmidt process, if we measure the first register and the result is $0$, we can  read out the $(k+1)$th state $|u_{t+1}\rangle$ from the second register as $|u_{t+1}\rangle\equiv|\psi\rangle$. So
	\begin{equation}
		a_{k+1}\in span\{u_1,u_2,\cdots,u_t,u_{t+1}\}
	\end{equation}
	and
	\begin{equation}
		\label{eqa2}
		span\{a_1,a_2,\cdots,a_k,a_{k+1}\} =span\{u_1,u_2,\cdots,u_t,u_{t+1}\}
	\end{equation}
	
	When $c$ as in Eq.~(\ref{eqa1}) is small, we may need to run the quantum circuit many times to measure 0 in the first register. When $a_{k+1}\in span\{a_1,a_2,\cdots,a_k\}$, then $c$ is 0 so we cannot measure $0$ from the first register. However, we cannot know this information in advance, so we need to run the quantum circuit repeatedly and make measurements to determine whether it is the linearly dependent case.
	
	Denote $p$ the probability of measuring $0$ in the first register. Then
	\begin{equation}
		p=c^2=1-\sum_{n=1}^{t}|\langle u_n|a_{k+1}\rangle|^2
	\end{equation}
	
	As we don't know anything about the distribution of $p$ in advance so the prior distribution of $p$ is uniform distribution over [0,1]. The first time 0 is measured in the first register, the measurement number is $W$. It is important to note that every measurement requires a repeated run of the  quantum circuit. Then $W$ is subject to geometric distribution.
	\begin{equation}
		P(W=w|p)=\left(1-p\right)^{w-1}p
	\end{equation}
	\begin{equation}
		P(W>w|p)=\left(1-p\right)^{w}
	\end{equation}
	So when 0 is still not measured after $w$ measurements, the posterior distribution of $p$ is
	\begin{align}
		P(p|W>w)&=\frac{P(W>w|p)P(p)}{\int_{0}^{1}P(W>w|p)P(p)dp}\nonumber\\
		&=\frac{(1-p)^w}{\int_{0}^{1}(1-p)^w dp}\nonumber\\
		&=(w+1)(1-p)^w
	\end{align}
	\indent So when the measurement number is large, but 0 is still not measured in the first register, then we know with high probability $a_{k+1}\in span\{a_1,a_2,\cdots,a_k\}$. \\
	\indent Define
	\begin{equation}
		f(x)=x\ln x-x\ln (1-x)+\ln x\ln (1-x), \;\;x\in(0,1).
	\end{equation}
	\indent Then
	\begin{equation}
		f'(x)=\left(-1+\frac{1}{x}\right)\ln\left(1-x\right)+\frac{-1+x\ln\left(x\right)}{-1+x}.
	\end{equation}
	\begin{equation}
		f''(x)=-\frac{\ln(1-x)+\frac{x(1-2x+x\ln(x))}{(-1+x)^2}}{x^2}.
	\end{equation}
	\indent Define
	\begin{equation}
		g(x)=\ln(1-x)+\frac{x(1-2x+x\ln(x))}{(-1+x)^2}.
	\end{equation}
	\indent Then
	\begin{equation}
		g'(x)=\frac{2x(x-\ln x)}{(-1+x)^{3}}.
	\end{equation}
	\indent As $x\in(0,1)$ and $x>\ln x$, so $g'(x)<0,\;\;x\in(0,1)$. And $\lim\limits_{x\to 0}g(x)=0$, so $g(x)<0,\;\;x\in(0,1)$. So $f''(x)>0,\;\;x\in(0,1)$.\\
	\indent As $\lim\limits_{x\to 0}f'(x)=0$, so $f'(x)>0,\;\;x\in(0,1)$. And $\lim\limits_{x\to 0}f(x)=0$ so
	\begin{equation}
		f(x)>0,\;\;x\in(0,1).
	\end{equation}
	
	Take $x=\epsilon$ we have
	\begin{equation}
		f(\epsilon)=\epsilon\ln \epsilon-\epsilon\ln(1-\epsilon)+\ln \epsilon\ln(1-\epsilon)>0.
	\end{equation}
	Devided by $\epsilon\ln(1-\epsilon)$. We get
	\begin{equation}
		\frac{\ln \epsilon}{\ln(1-\epsilon)}<\frac{1}{\epsilon}\ln(\frac{1}{\epsilon})+1.
	\end{equation}
	
	So when $w=1/\epsilon\ln\left(1/\epsilon\right)$,
	\begin{align}
		\label{eqa3}
		\int_{0}^{\epsilon}P(p|W>w)dp&=\int_{0}^{\epsilon} (w+1)(1-p)^wdp\nonumber\\
		&=1-(1-\epsilon)^{1+w}\nonumber\\
		&= 1-(1-\epsilon)^{\frac{1}{\epsilon}\ln\left(\frac{1}{\epsilon}\right)+1}\nonumber\\
		&>1-(1-\epsilon)^{\ln\epsilon/\ln (1-\epsilon)}\nonumber\\
		&=1-\epsilon.
	\end{align}
	Therefore, for each $k=1,2,\cdots,M$, if we run the quantum circuit of the $(k+1)$th step of quantum Gram-Schmidt process for $w=1/\epsilon\ln\left(1/\epsilon\right)$ times, and the measure outcome from the first register is $1$ for each run, then we can infer that $a_{k+1}\in span\{a_1,a_2,\cdots,a_k\}$ holds with probability larger than $1-\epsilon$. In this case, we do not construct a new $|u_{t+1}\rangle$ because it is already true that
	\begin{equation}
		span\{a_1,a_2,\cdots,a_{k+1}\}=span\{u_1,u_2,\cdots,u_{t}\}.
	\end{equation}
	
	To sum up, after the first $k$ steps, $span\{a_1,a_2,\cdots,a_k\}=span\{u_1,u_2,\cdots,u_t\}$ holds with probability $\Omega(1)$, then from Eq.~(\ref{eqa2}) we prove when we measure 0 from the first register then $span\{a_1,a_2,\cdots,a_{k+1}\}=span\{u_1,u_2,\cdots,u_{t+1}\}$ holds with probability $\Omega(1)\times1=\Omega(1)$. And from Eq.~(\ref{eqa3}) we prove when we measure the first register for $1/\epsilon\ln\left(1/\epsilon\right)$ times and still don't get 0, then $span\{a_1,a_2,\cdots,a_{k+1}\}=span\{u_1,u_2,\cdots,u_{t}\}$ holds with probability $\Omega(1)\times(1-\epsilon)=\Omega(1)$. Thus, after the first $k+1$ steps, the linear subspace spanned by the generated $\{u_t\}$ is equal to $span\{a_1,a_2,\cdots,a_{k+1}\}$ with probability $\Omega(1)$. The proof of Lemma \ref{lm6} is completed.
\end{proof}
\subsection{Proof of Lemma \ref{lm5}}
In Algorithm \ref{alg3}, we generate a series states successively, so we again use mathematical induction to prove Lemma \ref{lm5}.
\begin{proof}
	\indent Firstly, as we take $|u_1\rangle\equiv|a_1\rangle$, So $u_1$ itself is a normalized vector. It is clear that $|u_t\rangle$ and $|a_m\rangle$ is the amplitude encoding of $N-$dimension vector $u_t$ and $a_m$.
	\begin{equation}
		|a_m\rangle=\sum_{n=1}^{N}\frac{a_{nm}}{\parallel a_m \parallel} |n-1\rangle.
	\end{equation}
	\begin{equation}
		|u_t\rangle=\sum_{n=1}^{N}\frac{u_{nt}}{\parallel u_t \parallel} |n-1\rangle.
	\end{equation}
	
	Suppose in the first $k$ steps of quantum Gram-Schmidt process, we have generated $\{u_1,u_2,\cdots,u_t\}$ for $\{a_1,a_2,\cdots,a_k\}$, satisfying $u_{t_1}^{\dagger}u_{t_2}=O(\epsilon)$ with probability $\Omega(1)$ for $\forall t_1\neq t_2$. It is clear that $t$ is not always equal to $k$. It is because $\{a_1,a_2,\cdots,a_k\}$ is not always a set of linearly independent vectors so $k\geq t$. We want to prove the constructed $u_{t+1}$ is nearly  orthogonal to all previous constructed vectors $\{u_1,u_2,\cdots,u_t\}$ with the errant Hamiltonian simulation step of the QPE circuit, i.e.,
	\begin{equation}
		u_{t+1}^{\dagger}u_{t'}=O(\epsilon),\forall t'=1,2,\cdots,t
	\end{equation}
	
	Considering the $(k+1)$th step of quantum Gram-Schmidt process, we assume that $a_{k+1}\notin \{a_1,a_2,\cdots, a_k\}$. For the case that $a_{k+1}\in \{a_1,a_2,\cdots, a_k\}$, in this step we do not construct a new $|u_{t+1}\rangle$ so it is a trival case and follow the mathematical induction immediately. So we only consider the case that $a_{k+1}\notin \{a_1,a_2,\cdots, a_k\}$. In the $(k+1)$th step, the $\left(t+1\right)$th state $|u_{t+1}\rangle$ is constructed based on previous constrcuted states $\{|u_1\rangle,|u_2\rangle,\cdots,|u_{t}\rangle\}$ and $a_{k+1}$. $|a_{k+1}\rangle$ which can be written as
	\begin{equation}
		\label{eqa10}
		|a_{k+1}\rangle=\sum_{n=1}^{t}\langle u_n|a_{k+1}\rangle |u_n\rangle + c |\psi\rangle,
	\end{equation}
	where we denote $c$ to be a real number
	\begin{equation}
		c=\sqrt{1-\sum_{n=1}^{t}\parallel\langle u_n|a_{k+1}\rangle\parallel^2}.
	\end{equation}
	
	The ideal output of the circuit assuming there is no error in the Hamiltonian simulation step is
	\begin{align}
		|\psi\rangle_{ideal}&=U_{exact}\left(|0\rangle \otimes |a_{k+1}\rangle\right)\nonumber\\
		&=\sum_{n=1}^{t}\langle u_n|a_{k+1} \rangle|1\rangle |u_n\rangle + c|0\rangle |\psi\rangle
	\end{align}
	But with errant Hamilonian simulation step, the actual output state is $|\psi\rangle_{real}$.
	\begin{equation}
		|\psi\rangle_{real}=U_{real}\left(|0\rangle \otimes |a_{k+1}\rangle\right).
	\end{equation}
	
	From Lemma \ref{lm4} we know, the error of QPE circuit is bounded by
	\begin{equation}
		\parallel U_{real}-U_{exact}\parallel  < \epsilon_0.
	\end{equation}
	So
	\begin{align}
		\parallel |\psi\rangle_{real}-|\psi\rangle_{ideal}\parallel &=\parallel \left(U_{real}-U_{exact}\right)|0\rangle \otimes |a_{k+1}\rangle\parallel \nonumber\\
		& \leq \parallel U_{real}-U_{exact}\parallel \nonumber\\
		& < \epsilon_0.
		\label{eqa4}
	\end{align}
	\indent When we measure the first register and get result 0, then the state stored in the second register is then $|\psi\rangle$. We denote  $|\psi\rangle\equiv|u_{t+1}\rangle$.
	\begin{align}
		\frac{|0\rangle\langle 0|\otimes I |\psi\rangle_{ideal}}{\sqrt{\langle\psi|_{ideal}|0\rangle\langle 0|\otimes I|\psi\rangle_{ideal}}}&=|0\rangle \otimes |u_{t+1}\rangle_{ideal}\nonumber\\
		\frac{|0\rangle\langle 0|\otimes I |\psi\rangle_{real}}{\sqrt{\langle\psi|_{real}|0\rangle\langle 0|\otimes I|\psi\rangle_{real}}}&=|0\rangle \otimes |u_{t+1}\rangle_{real}
	\end{align}
	Denote the probability that 0 is measured from the first register $p(0)_{ideal}$ and $p(0)_{real}$ in the ideal case and in the real case, i.e.,
	\begin{equation}
		p(0)_{ideal}=\langle\psi|_{ideal}|0\rangle\langle 0|\otimes I|\psi\rangle_{ideal},
	\end{equation}
	\begin{equation}
		p(0)_{real}=\langle\psi|_{real}|0\rangle\langle 0|\otimes I|\psi\rangle_{real}
	\end{equation}
	So,
	\begin{align}
		&\;|p(0)_{real}-p(0)_{ideal}| \nonumber\\
		=& \;|\langle \psi|_{real}|0\rangle\langle 0|\otimes I|\psi\rangle_{real}-\langle\psi|_{ideal}|0\rangle\langle 0|\otimes I|\psi\rangle_{ideal}|\nonumber\\
		<& \;2\epsilon_0
	\end{align}
	
	As $\langle u_n|u_{t+1}\rangle_{ideal}=0,\;\forall n\leq t$,
	\begin{align}
		|\langle u_n|u_{t+1}\rangle_{real}|&= |\langle u_n|\left(|u_{t+1}\rangle_{real}-|u_{k+1}\rangle_{ideal}\right)|\nonumber\\
		&\leq \parallel |u_{t+1}\rangle_{real}-|u_{k+1}\rangle_{ideal}\parallel \nonumber\\
		&=\parallel |0\rangle|u_{t+1}\rangle_{real}-|0\rangle|u_{t+1}\rangle_{ideal}\parallel \nonumber\\
		&= \parallel \frac{|0\rangle\langle 0|\otimes I |\psi\rangle_{ideal}}{\sqrt{p(0)_{ideal}}}-\frac{|0\rangle\langle 0|\otimes I |\psi\rangle_{real}}{\sqrt{p(0)_{real}}}\parallel
	\end{align}
	\indent Expand $1/\sqrt{p(0)_{real}}$ as
	\begin{align}
		\frac{1}{\sqrt{p(0)_{real}}}&= \frac{1}{\sqrt{p(0)_{ideal}}}\frac{1}{\sqrt{1+\frac{p(0)_{real}-p(0)_{ideal}}{p(0)_{ideal}}}}\nonumber\\
		&= \frac{1}{\sqrt{p(0)_{ideal}}}\left(1+\frac{p(0)_{real}-p(0)_{ideal}}{2p(0)_{ideal}}\right)\nonumber\\
		&+O\left(\frac{\epsilon_0^2}{p(0)^{3/2}}\right)
	\end{align}
	\indent So,
	\begin{align}
		|\langle u_n|u_{t+1}\rangle_{real}| &\leq \parallel \frac{|0\rangle\langle 0|\otimes I \left(|\psi\rangle_{ideal}-|\psi\rangle_{real}\right)}{\sqrt{p(0)_{ideal}}}\parallel \nonumber\\
		& \;\;+ |\frac{p(0)_{real}-p(0)_{ideal}}{2p(0)_{ideal}}|*\parallel \frac{|0\rangle\langle 0|\otimes I |\psi\rangle_{real}}{\sqrt{p(0)_{ideal}}}\parallel \nonumber\\
		&\;\; +O\left(\frac{\epsilon_0^2}{p(0)^{3/2}}\right)*\parallel |0\rangle\langle 0|\otimes I |\psi\rangle_{real}\parallel\nonumber\\
		&< \frac{\epsilon_0}{\sqrt{p(0)_{ideal}}} + \frac{2\epsilon_0}{2*p(0)_{ideal}^{3/2}}+O\left(\frac{\epsilon_0^2}{p(0)^{3/2}}\right)\nonumber\\
		&<\frac{2\epsilon_0}{p(0)_{ideal}^{3/2}}
	\end{align}
	\indent For input $|a_{k+1}\rangle=\sum_{n=1}^{t}\langle u_n|a_{k+1}\rangle |u_n\rangle + c |\psi\rangle$
	\begin{align}
		p=p(0)_{ideal} = c^2= 1-\sum_{n=1}^{t}|\langle u_n|a_{k+1}\rangle|^2
	\end{align}
	\indent We dentoe the measurement time number $W$ when we measure 0 for the first time in the first register. Then
	\begin{equation}
		P(W=w|p)=\left(1-p\right)^{w-1}p
	\end{equation}
	\indent And the prior distribution of $p$ is uniform distribution over [0,1]. So posterior distribution of $p$ is
	\begin{align}
		P(p|W=w)&=\frac{P(W=w|p)P(p)}{\int_{0}^{1}P(W=w|p)P(p)dp}\nonumber\\
		&=\frac{\left(1-p\right)^{w-1}p}{\int_{0}^{1}\left(1-p\right)^{w-1}p \;dp}\nonumber\\
		&=w(w+1)p\left(1-p\right)^{w-1}
	\end{align}
	\indent Let $x=\delta^{1/2}/w$, so
	\begin{align}
		\int_{x}^{1}P(p|W=w)dp&=(1-x)^w(1+wx)\nonumber\\
		&>(1-wx)(1+wx)\nonumber\\
		&=1-w^2x^2\nonumber\\
		&=1-\delta
	\end{align}
	So we have
	\begin{equation}
		P(p>\frac{\delta^{1/2}}{w}|W=w)>1-\delta
	\end{equation}
	From section \ref{seca1} we know the maximum mearsurement number of the circuit is $1/\epsilon\ln\left(1/\epsilon\right)$. So with probability larger than $1-\delta$, we have
	\begin{equation}
		p > \frac{\delta^{1/2}}{w}> \frac{\delta^{1/2}}{\frac{1}{\epsilon}\ln\left(\frac{1}{\epsilon}\right)}
	\end{equation}
	So taking $\delta=\epsilon$, with probability larger than $1-\epsilon$
	\begin{equation}
		p > \frac{\epsilon^{1/2}}{w}> \frac{\epsilon^{1/2}}{\frac{1}{\epsilon}\ln\left(\frac{1}{\epsilon}\right)}.
	\end{equation}
	Therefore, with probability larger than $1-\epsilon$,
	
	\begin{align}
		|\langle u_n|u_{t+1}\rangle_{real}|&< \frac{2\epsilon_0}{p^{3/2}}\nonumber\\
		& < 2\epsilon_0\left(\frac{1}{\epsilon}\right)^{3/4}\left(\frac{1}{\epsilon}\ln\left(\frac{1}{\epsilon}\right)\right)^{3/2}\nonumber\\
		&=2\epsilon_0\left(\frac{1}{\epsilon}\right)^{9/4}\left(\ln\left(\frac{1}{\epsilon}\right)\right)^{3/2}\nonumber\\
		&=O\left(\epsilon_0\left(\frac{1}{\epsilon}\right)^{3}\right).
	\end{align}
	
	When we take $\epsilon_0=\epsilon^4$ as in Algorithm \ref{alg3}, we have our conclusion, With probability $\Omega(1)$, the constructed $|u_{t+1}\rangle$ with errant Hamiltonian simulation step is nearly orthogonal to all previous constructed states, i.e.,
	\begin{equation}
		|\langle u_n|u_{t+1}\rangle_{real}|=O(\epsilon), \;\; \forall n,\leq t.
	\end{equation}
	Therefore, in the $(k+1)th$ step of quantum Gram-Schmidt process, a state $|u_{t+1}\rangle$ is constructed, encoding a vector $u_{t+1}$. And $\{u_1,u_2,\cdots,u_t,u_{t+1}\}$ satisfies
	\begin{equation}
		u_{t_1}^{\dagger}u_{t_2}=O(\epsilon), \forall t_1\neq t_2
	\end{equation}
	with probability $\Omega(1)$. Thus, the proof of Lemma \ref{lm5} is completed.
\end{proof}



\bibliographystyle{apsrev4-1}
\bibliography{QR-decomposition}
\end{document}